\documentclass[12pt]{article}

\usepackage{amsmath,amssymb,amsfonts,amsthm}
\usepackage{algorithm}
\usepackage{algorithmic}
\usepackage{booktabs}
\usepackage{array}
\usepackage{tabularx}
\usepackage{multirow}
\usepackage{graphicx}
\usepackage{url}
\usepackage{xcolor}
\usepackage{enumitem}
\usepackage{float}
\usepackage{placeins}
\usepackage{geometry}
\usepackage{hyperref}
\usepackage{adjustbox}

\newtheorem{Theorem}{Theorem}
\newtheorem{Lemma}{Lemma}
\newtheorem{Corollary}{Corollary}
\newtheorem{Proposition}{Proposition}
\newtheorem{Definition}{Definition}
\newtheorem{Remark}{Remark}
\newtheorem{Example}{Example}

\allowdisplaybreaks[2]

\newcommand{\EJ}{\mathrm{EJ}}
\newcommand{\ZZ}{\mathbb{Z}}
\newcommand{\Norm}{\mathrm{N}}
\newcommand{\calS}{\mathcal{S}}
\newcommand{\calD}{\mathcal{D}}
\newcommand{\calP}{\mathcal{P}}

\newcommand{\red}{\mathrm{red}}

\newcommand{\vtx}[2]{\langle #1,#2\rangle}
\newcommand{\Arc}[3]{P_{#1}[#2>#3]}
\newcommand{\StepPair}[2]{\langle #1\,;\,#2\rangle}

\title{Ratio-Independent Three-Cycle Decomposition with Optimal Ordered Local-Switch Cost in Six-Regular Non-Axis Eisenstein--Jacobi Networks}

\author{Bader Albader\\
\small Department of Computer Science, College of Science, Kuwait University, P.O. Box 5969, Safat 13060, Kuwait\\
\small \texttt{albader@cs.ku.edu.kw}}

\date{}

\begin{document}
\maketitle

\begin{abstract}
Six-regular simple Eisenstein--Jacobi (EJ) networks are degree-six quotient-lattice interconnection networks. This paper gives a ratio-independent decomposition of every six-regular simple non-axis EJ network into three edge-disjoint Hamiltonian cycles using a canonical ordered local-switch model based on unit-parallelogram exchanges. The admitted $d=1$ branch needs no switches; $d=2$ has optimal total cost four; and for $d=3$ and $d\ge4$ both modified factors attain the component-counting lower bound $d-1$. Factor-local switches commute, so chronological interleaving does not alter the final factors or cost within the model. Orbit normalization identifies the exact domain and excludes the unique normalized non-axis norm-three degeneration. For $d\ge4$, an equal-coordinate alternating lift removes reduced-ratio dependence from the fine diagonal coordinate. A block-chain invariant, exhaustive interior-template lemma, and parity-specific successor permutations certify the unused complement: rank advances by one modulo $4d-6$, and arc and connector bijections prove complete coverage. The certificate uses $O(d)$ seed records and expands to the full edge lists in $O(N)$ time. Deterministic symbolic and full-quotient audits, including a dictionary-free fine-incidence check for every $4\le d\le201$, are provided in the accompanying reproducibility package and are not proof premises.
\end{abstract}


\section{Introduction}

Hamiltonian decompositions are useful in interconnection networks because a spanning cycle supplies a simple ring for broadcasting, all-to-all communication, token circulation, diagnostics, and fault-tolerant routing. A degree-six network can contain at most three edge-disjoint Hamiltonian cycles (EDHCs), so a decomposition into three EDHCs is optimal: it covers every edge exactly once. Six-regular EJ networks are especially relevant because their hexagonal quotient geometry provides three symmetric routing directions in a compact algebraic model, making them a natural alternative to two-axis toroidal topologies for communication structures that exploit all available direction classes \cite{flahive2010,martinez2008ej}. The contribution below is graph-theoretic, but its compact certificate has a concrete systems interpretation: a controller can store or generate $O(d)$ switch-seed records, verify them locally, and expand them into the full $O(N)$ edge lists only when required. No runtime speedup over an existing implementation is claimed.

Let $\alpha=a+b\rho$, where $\rho=(1+i\sqrt3)/2$, and let $\EJ_\alpha$ be the quotient Cayley graph over $\ZZ[\rho]/(\alpha)$ with connection set $\{\pm1,\pm\rho,\pm\rho^2\}$. Every quotient has $N=\Norm(\alpha)=a^2+ab+b^2$ vertices. On the six-regular simple domain treated in this paper, the six signed unit residues are distinct, so the graph has $3N$ undirected edges. Section~\ref{sec:normalization} identifies the exact boundary: besides the separately excluded axis-associated family, the unique normalized non-axis degeneration is $\alpha=1+\rho$, whose underlying simple graph is $C_3$. Write $d=\gcd(a,b)$. On the admitted $d=1$ branch the three unit directions are Hamiltonian cycles; if $d>1$, the natural direction factors must be spliced.

The general non-coprime existence problem was solved previously through a rectangular representation with $d$ rows and $r=N/d$ columns \cite{hussain2015}. The present paper does not re-claim that result. Instead, it asks how a fixed $O(d)$ local switch skeleton can replace a coordinate-indexed rectangular schedule, attain the minimum total local-switch count, and certify the Hamiltonicity of the unused complement. The phrase \emph{reduced-ratio case analysis} below means dividing the normalized coprime pairs $(u,v)$ into separate arithmetic families before choosing the lift. Throughout this paper, \emph{ratio-independent} means that no such family classification or seed search is required: the exceptional $d=2$ seeds are evaluated directly from the input parameters, while the universal $d\ge4$ lift phases themselves are independent of $(u,v)$.

The main contributions are as follows.
\begin{itemize}[leftmargin=*]
 \item The ordered cost model is tied to the canonical direction factorization and the three unit-parallelogram exchanges; up to direction symmetry these are all elementary two-edge exchanges of this type. Switches assigned to different factor copies commute, so chronological interleaving loses no generality inside the model.
 \item The three natural direction factors are shown to have exactly $d$ components, and a component argument gives the lower bound $d-1$ intercomponent switches for a Hamiltonian factor.
 \item Explicit formulas cover every coprime reduced ratio for $d=2$ and $d=3$. Both branches have proved minimum total cost four; the improved $d=3$ construction uses two switches in each of the first two factors.
 \item For $d\ge4$, horizontal and vertical switch skeletons attain the lower bound, while anchor lock removes their only forced non-diagonal overlap.
 \item An equal-coordinate alternating lift makes the fine phase sequence $0,1,0,1,\ldots$ independently of $(u,v)$; separate even and odd seam formulas then certify every reduced ratio.
 \item A successor-permutation rank certificate proves that the complement is one Hamiltonian cycle: every diagonal arc and every available connector is used exactly once, and rank advances by one modulo $4d-6$.
 \item The resulting certificate has $O(d)$ seed-record size and expands to all three cycles in output-optimal $O(N)$ time. Deterministic audits are supplied as verification, not proof.
\end{itemize}

\section{Relation to the Rectangular Construction}

The quotient-ring topology of Gaussian and EJ networks was developed by Flahive and Bose \cite{flahive2010}, with related EJ constellation models given in \cite{martinez2008ej}. The closest prior EDHC result is due to Hussain, Bose, and Al-Dhelaan \cite{hussain2015}. They arrange vertices as $i\rho^2+k$, with $0\le i<d$ and $0\le k<r=N/d$, and exchange selected horizontal, vertical, and diagonal edges. Their construction is a valid general existence proof.

Recent EJ-specific work has developed independent spanning trees, panconnectivity algorithms, and completely independent spanning trees for routing, broadcasting, and fault tolerance \cite{hussain2022ist,awadh2023pan,hussain2024cist}. Edge-disjoint Hamiltonian structures also remain active in general Cayley graphs, balanced hypercubes, BCube networks, and Gaussian quotient networks \cite{yang2023cayley,cheng2024balanced,pai2026bcube,albader2026gaussian}. These studies confirm continuing interest in resilient EJ and algebraic-network structures. Within EJ networks, however, they do not provide the ratio-independent, minimum-cost local certificate studied here.

Among these works, the closest methodological analogue is the author's recent degree-four Gaussian-network construction \cite{albader2026gaussian}, which also uses local switches and a compact certificate. The present EJ result should therefore be viewed as an extension of that local-switch program, rather than as a new existence theorem. It is nevertheless not obtained by merely replacing Gaussian integers with Eisenstein--Jacobi integers: a degree-six decomposition must coordinate two modified direction factors before certifying a third unused complement. This creates overlap constraints between $H_1$ and $H_2$, multiple fine arcs above one diagonal residue, parity-dependent terminal seams, and a $4d-6$-state successor certificate. The anchor-lock mechanism, optimal exceptional $d=2,3$ branches, equal-coordinate ratio-cancelling lift, and parity-complete rank maps are the EJ-specific ingredients developed here.

Related decomposition and Gray-code constructions have been studied for Gaussian networks and Gaussian-integer tori \cite{albader2016,martinez2008gaussian}, $k$-ary cubes, hypercubes, and toroidal networks \cite{bae2003,bae2000,bae2004,latifi1993,jha2012}, De Bruijn networks \cite{rowley1991,rowley1993}, and mixed-radix Lee-metric codes \cite{anantha2007}. General context on Hamiltonian Cayley graphs and decompositions appears in \cite{chen1981,witte1984,bermond1989,alspach2008}; standard background on interconnection networks, number theory, and algebraic graph theory is provided by \cite{duato2003,grama2003,hardy1980,biggs1993,godsil2001,bondy2008}.

The present framework works directly in the Cayley lattice and targets $2(d-1)$ intercomponent switches for the first two factors. Its compact certificate contains $O(d)$ seed records, whereas the rectangular schedule is indexed over a coordinate of length $\Theta(r)=\Theta(N/d)$. These are different descriptive currencies: \cite{hussain2015} did not optimize or report the present local-switch count, and no claim is made here that one implementation executes faster or is simpler in every engineering setting. The precise comparison is that \cite{hussain2015} proves general existence through a coordinate-indexed rectangular schedule, while the present paper gives, on the admitted six-regular non-axis domain, an $O(d)$ local-seed certificate, proves the minimum value $d-1$ for each of the first two factors, and retains the unavoidable $O(N)$ cost when all cycle edges must be listed explicitly.

\begin{table}[H]
\centering
\caption{Certificate-description comparison with the rectangular construction. The example uses $\alpha=12+30\rho$, for which $d=6$, $r=N/d=234$, and $N=1404$. The methods use different primitives; this table does not compare measured running times or claim that \cite{hussain2015} optimized the present switch metric.}
\label{tab:method-comparison}
\small
\begin{adjustbox}{max width=\textwidth}
\begin{tabularx}{1.3\textwidth}{@{}p{0.18\textwidth}XXXX@{}}
\toprule
Method & Compact control description & Local-switch certificate & Example control scale & Full edge listing\\
\midrule
Rectangular construction \cite{hussain2015} & schedule indexed over $r=\Theta(N/d)$ columns & not formulated or minimized in the present switch metric & coordinate span $r=234$ & $\Theta(N)$ output\\
Present construction & $O(d)$ explicit local seeds & $2(d-1)$ total for $H_1,H_2$, with $d-1$ minimum for each & $2(d-1)=10$ switches & $O(N)$ time\\
\bottomrule
\end{tabularx}
\end{adjustbox}
\end{table}

The geometric comparison and the full positioning diagram are reproduced as Supplementary Figure~S1.

\begin{table}[H]
\centering
\caption{Core notation used in the universal construction. Exceptional $d=2,3$ notation is introduced locally in Section~\ref{sec:exceptional}.}
\label{tab:notation}
\footnotesize
\begin{tabularx}{0.8\textwidth}{@{}lX@{}}
\toprule
Symbol & Meaning\\
\midrule
$\alpha=a+b\rho$ & EJ generator\\
$N=a^2+ab+b^2$ & number of vertices\\
$d=\gcd(a,b)$ & number of cycles per unit direction\\
$r=N/d$ & order of each unit direction\\
$(u,v)=(a/d,b/d)$ & reduced coprime ratio\\
$s_i^{\red}=(\xi_i,\eta_i)$ & reduced switch seed\\
$s_i=s_i^{\red}+d(\pi_i,\kappa_i)$ & lifted seed\\
$z=x-y\pmod d$ & coarse diagonal-cycle label\\
$P_{z,\lambda}$ & actual diagonal arc segment\\
$\vtx{z}{t}$ & unique quotient vertex with fine coordinates $(z,t)$\\
$\Gamma(\Phi)$ & fine complement-incidence graph\\
\bottomrule
\end{tabularx}
\end{table}

\section{Orbit Normalization and Admitted Domain}\label{sec:normalization}

Let $\mathcal U=\{\pm1,\pm\rho,\pm\rho^2\}$ be the unit group of $\ZZ[\rho]$.

\begin{Definition}[Axis-associated generator]\label{def:axis}
A nonzero generator $\alpha\in\ZZ[\rho]$ is \emph{axis-associated} if $\alpha=n\varepsilon$ for some $n\in\ZZ\setminus\{0\}$ and $\varepsilon\in\mathcal U$. A generator that is not axis-associated is called \emph{non-axis}. A coordinate representative $a+b\rho$ is \emph{normalized} if $0\le a\le b$.
\end{Definition}

\begin{Lemma}[Dihedral orbit normalization]\label{lem:orbit-normalization}
Let $\alpha=a+b\rho\ne0$. Among the twelve elements
\begin{equation}
 \mathcal O(\alpha)=\{\varepsilon\alpha,\varepsilon\overline\alpha:\varepsilon\in\mathcal U\}
\label{eq:dihedral-orbit}
\end{equation}
there is a representative $A+B\rho$ with $0\le A\le B$. Moreover,
\begin{equation}
 A=0\quad\Longleftrightarrow\quad \alpha\text{ is axis-associated}
 \quad\Longleftrightarrow\quad ab(a+b)=0.
\label{eq:axis-characterization}
\end{equation}
Every map used in \eqref{eq:dihedral-orbit} induces a graph isomorphism from $\EJ_\alpha$ to the EJ graph generated by the transformed element. Consequently every non-axis generator has an isomorphic normalized representative with $0<A\le B$.
\end{Lemma}

\begin{proof}
In coefficient coordinates, multiplication by $\rho$ and conjugation are
\begin{equation}
 R(x,y)=(-y,x+y),\qquad C(x,y)=(x+y,-y).
\label{eq:orbit-matrices}
\end{equation}
Thus \eqref{eq:dihedral-orbit} is the finite set $\{R^j(a,b),R^jC(a,b):0\le j<6\}$. Both matrices are unimodular, so $\gcd(a,b)$ is invariant throughout the orbit. Geometrically, $R$ is rotation by $60^\circ$ and $C$ is reflection. Their dihedral action partitions the plane into twelve closed $30^\circ$ chambers, and $\mathcal K=\{A+B\rho:0\le A\le B\}$ is one such chamber. Hence every nonzero orbit meets $\mathcal K$.

The six unit-direction axes are exactly the three coefficient lines $a=0$, $b=0$, and $a+b=0$ with both orientations. Therefore $\alpha$ is a unit multiple of a nonzero integer exactly when $ab(a+b)=0$. These axis lines are mapped by the dihedral action to the boundary $A=0$ of $\mathcal K$; the other chamber boundary $A=B$ is a reflection bisector and is not a unit-direction axis. This proves \eqref{eq:axis-characterization} and shows that a non-axis orbit representative in $\mathcal K$ has $A>0$.

Finally, multiplication by a unit sends the ideal $(\alpha)$ to $(\varepsilon\alpha)$ and permutes $\{\pm1,\allowbreak\pm\rho,\allowbreak\pm\rho^2\}$; conjugation sends $(\alpha)$ to $(\overline\alpha)$ and also permutes this connection set. Both operations therefore induce quotient Cayley-graph isomorphisms and preserve Hamiltonian decompositions.
\end{proof}

\begin{Remark}[Deterministic normalization]\label{rem:det-normalization}
Algorithm~\ref{alg:certified} evaluates the twelve explicit pairs in \eqref{eq:orbit-matrices}, selects any pair with $0\le A\le B$ (lexicographic tie-breaking may be used), and first tests whether $A=0$. Lemma~\ref{lem:orbit-normalization} proves that this constant-size search always terminates and that $A=0$ is an exact, checkable test for the excluded axis-associated class. Proposition~\ref{prop:nonaxis-degenerate} supplies the second exact test, $(A,B)=(1,1)$, for the unique non-axis quotient that is not six-regular simple.
\end{Remark}

\begin{Proposition}[Exact axis boundary]\label{prop:axis-boundary}
If $\alpha$ is axis-associated, its normalized representative is $n\rho$ or, after a unit rotation, the integer generator $n$ with $n\ge1$. The additive quotient is $\ZZ_n^2$. For $n=1$ the quotient has one vertex, and for $n=2$ the six signed unit steps collapse to three distinct neighbors in the underlying simple graph. For $n\ge3$ the quotient is the six-regular triangular torus
\begin{equation}
 \operatorname{Cay}\bigl(\ZZ_n^2,\{\pm H,\pm V,\pm D\}\bigr).
\label{eq:axis-torus}
\end{equation}
Its reduced pair is $(u,v)=(0,1)$ and hence $Q=1$, so the distinct phase representatives $1$ and $2$ required by the alternating lift coincide modulo $Q$. Thus the non-axis construction cannot be specialized to the axis family by a limiting substitution.
\end{Proposition}

\begin{proof}
The orbit characterization in Lemma~\ref{lem:orbit-normalization} reduces an axis-associated generator to a unit multiple of $n$. Reduction modulo $(n)$ gives two coefficient coordinates modulo $n$, hence $\ZZ_n^2$. The claims for $n=1,2$ follow by reducing the six signed steps; for $n\ge3$ they remain six distinct nonzero group elements and generate the triangular torus. Normalization of the integer generator has one zero reduced coordinate, which gives $Q=0^2+0\cdot1+1^2=1$ and collapses all phase residues.
\end{proof}

\begin{Proposition}[Exact non-axis simple-degree boundary]\label{prop:nonaxis-degenerate}
Let $A+B\rho$ be a normalized non-axis generator, so $0<A\le B$. The six signed unit residues $\{\pm1,\pm\rho,\pm\rho^2\}$ are pairwise distinct in $\ZZ[\rho]/(A+B\rho)$ if and only if $(A,B)\ne(1,1)$, equivalently $\Norm(A+B\rho)\ge7$. For $(A,B)=(1,1)$ the quotient has order three and its underlying simple graph is $C_3$.
\end{Proposition}

\begin{proof}
If two distinct signed units $\varepsilon,\delta\in\mathcal U$ have the same residue, then $A+B\rho$ divides $\varepsilon-\delta$. A direct check of the six units gives
\[
 \Norm(\varepsilon-\delta)\in\{1,3,4\}.
\]
Hence $\Norm(A+B\rho)$ divides one of $1,3,4$. Since $0<A\le B$, one has $\Norm(A+B\rho)=A^2+AB+B^2\ge3$. The equation $A^2+AB+B^2=3$ has the unique positive normalized solution $(A,B)=(1,1)$, while the value $4$ has no positive normalized solution. Thus every normalized non-axis pair other than $(1,1)$ has six distinct signed unit residues; none is zero because a nonunit generator cannot divide a unit.

For $A=B=1$, the relation $1+\rho\equiv0$ gives $\rho\equiv-1$ and $\rho^2=\rho-1\equiv1$ in the quotient. The three positive directions therefore collapse to the single undirected direction $\{\pm1\}$ on a group of order three, so the underlying simple graph is $C_3$. Finally, if $(A,B)\ne(1,1)$ then $A^2+AB+B^2>3$; the next possible value in the positive normalized chamber is at least $7$, proving the equivalent norm test.
\end{proof}

\begin{Remark}[Why the theorem is stated for six-regular non-axis networks]\label{rem:axis-scope}
Propositions~\ref{prop:axis-boundary} and \ref{prop:nonaxis-degenerate} give the complete excluded boundary. In the axis-associated family, $n=1,2$ are not six-regular simple EJ networks and $n\ge3$ is a separate triangular-torus family whose phase space has only one residue. Within the non-axis chamber, the sole nonsix-regular simple quotient is the norm-three case $1+\rho$, whose underlying simple graph is $C_3$. A local minimum-switch treatment of the axis family is not claimed, and the collapsed norm-three graph cannot contain three edge-disjoint Hamiltonian cycles. Accordingly, the title, abstract, theorem, and algorithm state the exact six-regular non-axis domain rather than hiding either boundary as an exception.
\end{Remark}

\section{Coordinates and Direction Factors}

Use the basis $\{1,\tau\}$, where $\tau=\rho-1=\rho^2$ and $\rho=1+\tau$. A vertex is written as $(x,y)=x+y\tau$, and the three positive direction steps are $H=(1,0)$, $V=(0,1)$, and $D=(1,1)$.

\begin{figure}[H]
 \centering
 \includegraphics[width=0.6\linewidth]{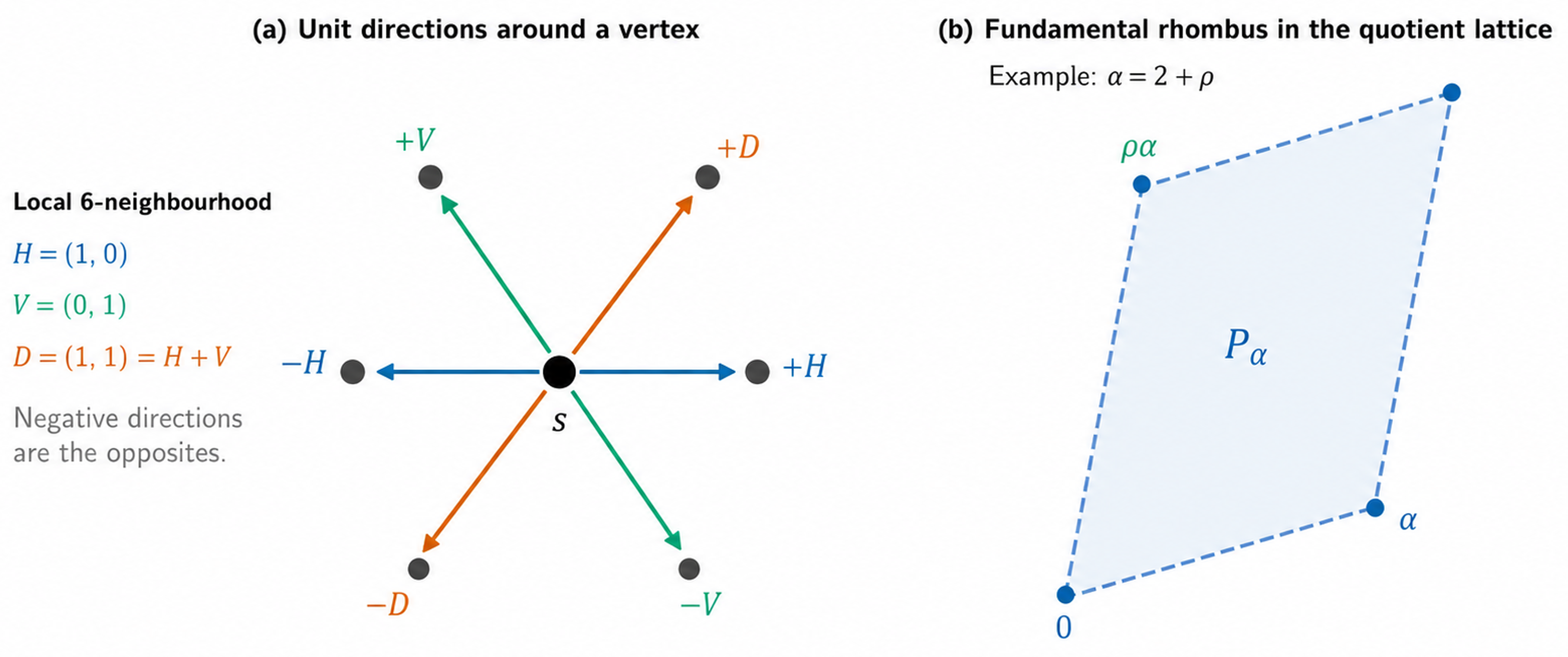}
 \caption{The three EJ directions in the $\{1,\tau\}$ coordinate system and the fundamental quotient rhombus.}
 \label{fig:directions}
\end{figure}

\begin{Lemma}[Orders and component labels]\label{lem:direction-orders}
Let $a=du$, $b=dv$, and $\gcd(u,v)=1$. In the $\{1,\tau\}$ coordinates, each of $H,V,D$ has additive order
\begin{equation}
 r=d(u^2+uv+v^2)=N/d
\end{equation}
in $\ZZ[\rho]/(\alpha)$. Consequently, every natural direction factor consists of exactly $d$ cycles of length $r$. Their component labels can be taken as $H:y\pmod d$, $V:x\pmod d$, and $D:x-y\pmod d$.
\end{Lemma}

\begin{proof}
In the $\{1,\tau\}$ basis, $\alpha=d(u+v,v)$ and $\tau\alpha=d(-v,u)$, so these two vectors generate the quotient lattice. Let $Q=u^2+uv+v^2$.

For $H=(1,0)$, an equality $nH=m\alpha+k\tau\alpha$ gives $mv+ku=0$. Hence $m=ut$, $k=-vt$, and $n=dQt$. Thus the order of $H$ is $dQ$. For $V=(0,1)$, the first coordinate gives $m(u+v)=kv$, so $m=vt$, $k=(u+v)t$, again yielding $n=dQt$. For $D=(1,1)$, equality of the two coordinates gives $mu=k(u+v)$, so $m=(u+v)t$, $k=ut$, and again $n=dQt$. Therefore all three directions have order $r=dQ=N/d$, and each 2-factor has $N/r=d$ components.

The displayed residues are preserved by the corresponding step and are identified modulo $d$ by the quotient lattice. Since there are exactly $d$ components, these residues label them completely.
\end{proof}

\begin{Corollary}[Admitted coprime branch]\label{cor:d1}
Assume $d=1$ and $\EJ_\alpha$ is six-regular simple. Then each natural direction factor is a single Hamiltonian cycle of length $N$; hence the three direction factors already form an edge decomposition into three EDHCs.
\end{Corollary}

\begin{proof}
Lemma~\ref{lem:direction-orders} gives direction order $r=N/d=N$ and exactly $N/r=1$ component in each direction. Proposition~\ref{prop:nonaxis-degenerate} ensures that the three undirected direction edge sets are distinct on the admitted non-axis domain. They are therefore pairwise disjoint and together form the complete degree-six edge set.
\end{proof}

\begin{Remark}[Basis conversion]
The coordinate $D=(1,1)$ is a $\{1,\tau\}$-basis statement. Since $(x,y)_\tau=(x-y,y)_\rho$, one has $D=(1,1)_\tau=(0,1)_\rho=\rho$. Using $(1,1)$ in the $\rho$ basis would instead mean $1+\rho$, a different group element.
\end{Remark}

\section{Local Switch Calculus}

Each switch is supported on the unit parallelogram spanned by the direction being removed and the direction being inserted. The three switch types therefore use different parallelograms.

\begin{Definition}[Valid local switches]\label{def:switches}
At a seed vertex $s$:
\begin{itemize}[leftmargin=*]
 \item $H\leftarrow V$ uses $\{s,s+H,s+V,s+H+V\}$, removes
 $[s,s+H]$ and $[s+V,s+H+V]$, and inserts
 $[s,s+V]$ and $[s+H,s+H+V]$.
 \item $H\leftarrow D$ uses $\{s,s+H,s+D,s+H+D\}$, removes
 $[s,s+H]$ and $[s+D,s+H+D]$, and inserts
 $[s,s+D]$ and $[s+H,s+H+D]$.
 \item $V\leftarrow D$ uses $\{s,s+V,s+D,s+V+D\}$, removes
 $[s,s+V]$ and $[s+D,s+V+D]$, and inserts
 $[s,s+D]$ and $[s+V,s+V+D]$.
\end{itemize}
\end{Definition}

\begin{figure}[H]
 \centering
 \includegraphics[width=0.65\linewidth]{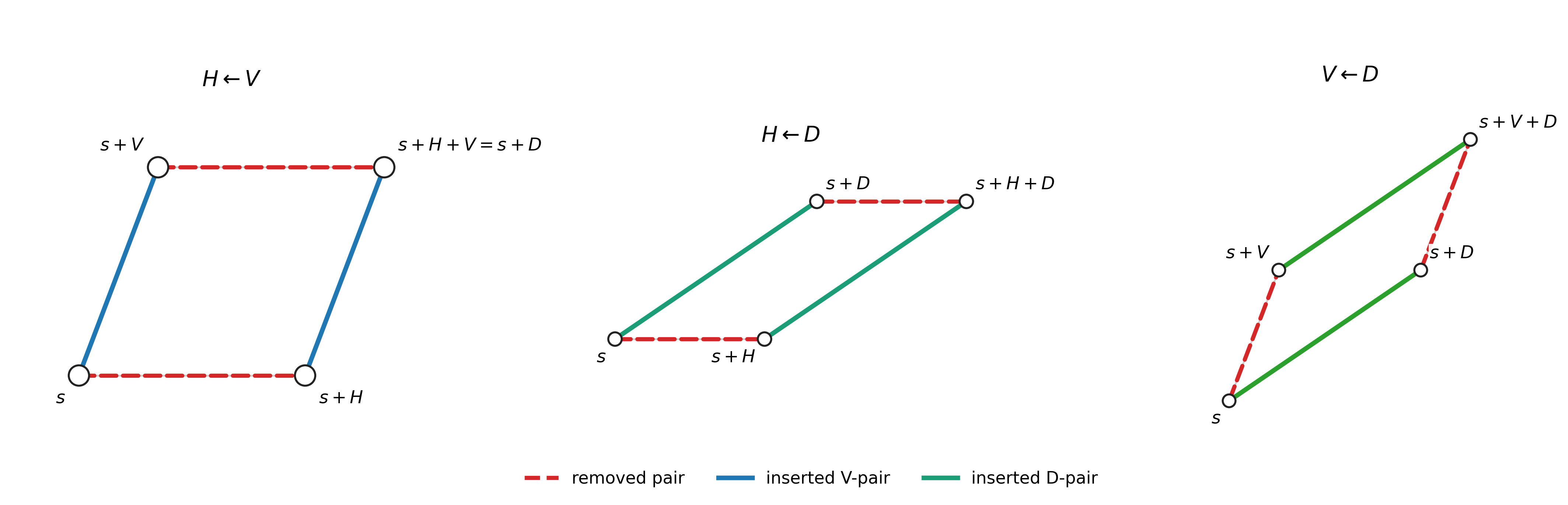}
 \caption{The three geometrically valid switch parallelograms. In particular, the $H\leftarrow D$ and $V\leftarrow D$ switches are not drawn as crossing diagonals of the $H$--$V$ rhombus.}
 \label{fig:switches}
\end{figure}

\begin{Lemma}[Component action]\label{lem:component-action}
An $H\leftarrow V$ or $H\leftarrow D$ switch at $(x,y)$ joins horizontal labels $y$ and $y+1$. A $V\leftarrow D$ switch at $(x,y)$ joins vertical labels $x$ and $x+1$. Every switch preserves degree two.
\end{Lemma}

\begin{proof}
The two removed parallel edges belong to the two stated direction components. The inserted opposite pair reconnects their four loose endpoints crosswise. Each involved vertex loses one incident edge and gains one incident edge, so the modified factor remains 2-regular. If the two labels are distinct, the switch merges the corresponding components.
\end{proof}

\begin{Lemma}[Switch lower bound]\label{lem:lowerbound}
Starting from a direction factor with $d$ components, any sequence of local switches that produces one Hamiltonian cycle must contain at least $d-1$ intercomponent switches.
\end{Lemma}

\begin{proof}
A switch can reduce the number of connected components of a 2-factor by at most one. Reducing $d$ components to one therefore requires at least $d-1$ switches.
\end{proof}

\begin{Proposition}[Canonical unit-parallelogram exchanges]\label{prop:canonical-switches}
Let $F_X$ be one of the three natural direction factors, where $X\in\{H,V,D\}$. A degree-preserving two-edge exchange supported on a unit parallelogram spanned by two distinct EJ directions removes the two opposite $X$-edges and inserts the two opposite edges in the other direction. Up to permutation and reversal of the three directions, every such exchange is one of the three switches in Definition~\ref{def:switches}.
\end{Proposition}

\begin{proof}
A unit parallelogram is determined by an unordered pair from $\{H,V,D\}$. There are exactly three such pairs: $\{H,V\}$, $\{H,D\}$, and $\{V,D\}$. On each parallelogram, preserving degree two while replacing the two opposite edges of one direction forces the insertion of the two opposite edges of the other direction. Choosing which member is the current factor gives precisely the three representative exchanges in Definition~\ref{def:switches}; all remaining orientations are obtained by permuting or reversing directions.
\end{proof}

\begin{Definition}[Canonical ordered local-switch model and cost]\label{def:local-cost}
Up to permutation of the three unit directions, an \emph{ordered local-switch decomposition} starts with the natural $H$ and $V$ factors, modifies the first by switches of type $H\leftarrow V$ or $H\leftarrow D$, modifies the second by switches of type $V\leftarrow D$, and takes the third factor to be the unused complement. Its cost is the total number of unit-parallelogram exchanges. Translation, unit rotation, conjugation, reversal of a factor, and interchange of the first two factors do not change this cost.
\end{Definition}

\begin{Lemma}[Commutation of factor-local switches]\label{lem:switch-commutation}
Consider any finite sequence of the switches allowed by Definition~\ref{def:local-cost}, with every switch assigned to either the $H$-derived factor copy or the $V$-derived factor copy. The sequence can be stably reordered so that all switches assigned to the first copy precede all switches assigned to the second copy, without changing either final factor or the total cost.
\end{Lemma}

\begin{proof}
A switch assigned to the $H$-derived copy reads and changes only that factor's edge set, and a switch assigned to the $V$-derived copy reads and changes only the other edge set. Hence two adjacent switches assigned to different copies commute as operations on the ordered pair of factor edge sets. Repeatedly swapping adjacent switches of different assignments performs a stable sort: the internal order of the switches on each individual factor is preserved, so their validity and their final factor are unchanged. The third factor is formed only after the two final factors are fixed, as their unused complement. Thus chronological interleaving has no effect on the resulting decomposition or on the number of exchanges.
\end{proof}

\begin{Remark}[Meaning and limitation of the cost model]\label{rem:model-scope}
Definition~\ref{def:local-cost} measures local exchange distance from the canonical three-direction factorization. Lemma~\ref{lem:switch-commutation} shows that the word ``ordered'' fixes only a normal form for the chronology of the permitted factor-local operations; arbitrary interleaving of those same operations reaches exactly the same final factor pairs at the same cost. Once two disjoint spanning 2-factors have been selected in a six-regular graph, the third factor is forced as their complement. The resulting cost is therefore a natural invariant of this local construction model, but it is not claimed to be a graph invariant over broader models that allow direct modification of all three factors, cross-factor exchanges, or arbitrary global encodings of Hamiltonian decompositions. All optimality statements below are confined to the explicitly stated model.
\end{Remark}

\section{Minimal-Alignment Skeleton}

The first factor $H_1$ is built from the horizontal factor by
\begin{equation}
\calS_H(d)=\{(H\leftarrow V,(0,0))\}\cup
\{(H\leftarrow D,(d-1,q)):q=d-1,d-2,\ldots,2\}.
\label{eq:H-skeleton}
\end{equation}
It contains $d-1$ switches.

The second factor starts from $V$ and uses $V\leftarrow D$ switches. Its reduced seed set is, for even $d$,
\begin{equation}
\calS_V^{\red}(d)=\{(-1,-1),(1,0),(2,0),\ldots,(d-2,0)\},
\label{eq:V-even}
\end{equation}
and, for odd $d$,
\begin{equation}
\calS_V^{\red}(d)=\{(-1,-1),(1,0),(2,0),\ldots,(d-3,0),(d-2,-1)\}.
\label{eq:V-odd}
\end{equation}

\begin{figure}[H]
 \centering
 \includegraphics[width=0.65\linewidth]{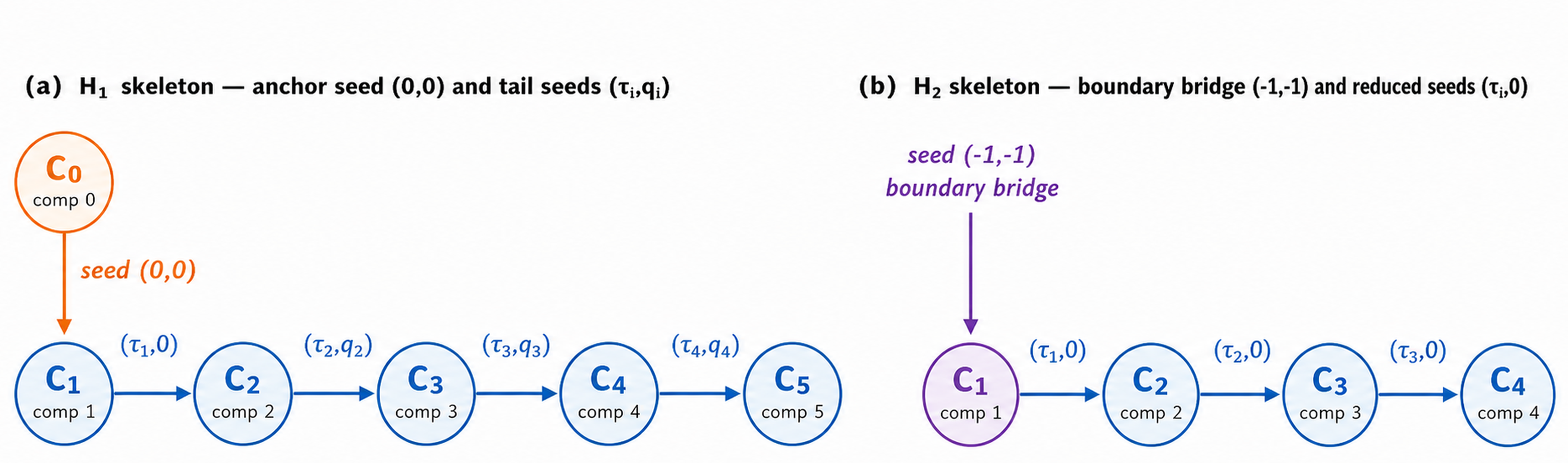}
 \caption{Quotient-level minimal-alignment skeleton. The first two direction factors collapse their $d$ components using the lower-bound number $d-1$ of intercomponent switches.}
 \label{fig:skeleton}
\end{figure}

\begin{Lemma}[Minimal collapse of $H_1$]\label{lem:H1-collapse}
The switch list \eqref{eq:H-skeleton} induces a spanning tree on the $H$-component labels. Every valid realization of these switches therefore produces a Hamiltonian cycle $H_1$.
\end{Lemma}

\begin{proof}
The $H\leftarrow V$ seed $(0,0)$ joins labels $0$ and $1$. The $H\leftarrow D$ seeds $(d-1,q)$ join $q$ to $q+1$ for $q=d-1,d-2,\ldots,2$ modulo $d$. These $d-1$ edges form a spanning tree on $\ZZ_d$. By Lemma~\ref{lem:component-action}, each switch merges two components while preserving degree two.
\end{proof}

\begin{Lemma}[Minimal collapse of $H_2$]\label{lem:H2-collapse}
The reduced seeds \eqref{eq:V-even} and \eqref{eq:V-odd} induce a spanning tree on the $V$-component labels. Any lift whose switch bases are distinct therefore produces a Hamiltonian cycle $H_2$.
\end{Lemma}

\begin{proof}
A $V\leftarrow D$ seed with first coordinate $x$ joins labels $x$ and $x+1$. The first-coordinate residues in the even case are $d-1,1,2,\ldots,d-2$, producing the spanning tree $\{d-1,0\},\{1,2\},\{2,3\},\ldots,\{d-2,d-1\}$. The odd list has the same first-coordinate residues; its final $y$ coordinate changes the seam representative but not the merged labels.
\end{proof}

\section{Lift Phases and Anchor Lock}

Each reduced seed has representatives
\begin{equation}
 s_i=(\xi_i,\eta_i)+d(\pi_i,\kappa_i),\qquad \pi_i,\kappa_i\in\ZZ.
\label{eq:lift}
\end{equation}
The pair $(\pi_i,\kappa_i)$ is the lift phase. It changes the location of a switch on the quotient seam without changing its reduced component labels.

\begin{Definition}[Anchor lock]\label{def:anchor}
A lift satisfies anchor lock if the two reduced seeds $(-1,-1)$ and $(1,0)$ are represented so that their $V\leftarrow D$ switches remove the vertical bases $(0,0)$ and $(1,0)$ inserted by the $H\leftarrow V$ switch of $H_1$. The canonical representatives are $(-1,-1)\mapsto(-1,-1)$ and $(1,0)\mapsto(1,0)$.
\end{Definition}

\begin{Lemma}[Anchor disjointness]\label{lem:anchor}
Under anchor lock, the two vertical edges inserted into $H_1$ at the anchor do not belong to $H_2$.
\end{Lemma}

\begin{proof}
The $H\leftarrow V$ switch at $(0,0)$ inserts the vertical bases $(0,0)$ and $(1,0)$. The $V\leftarrow D$ switch at $(-1,-1)$ removes vertical bases $(-1,-1)$ and $(0,0)$, while the switch at $(1,0)$ removes $(1,0)$ and $(2,1)$.
\end{proof}

\begin{figure}[H]
 \centering
 \includegraphics[width=0.65\linewidth]{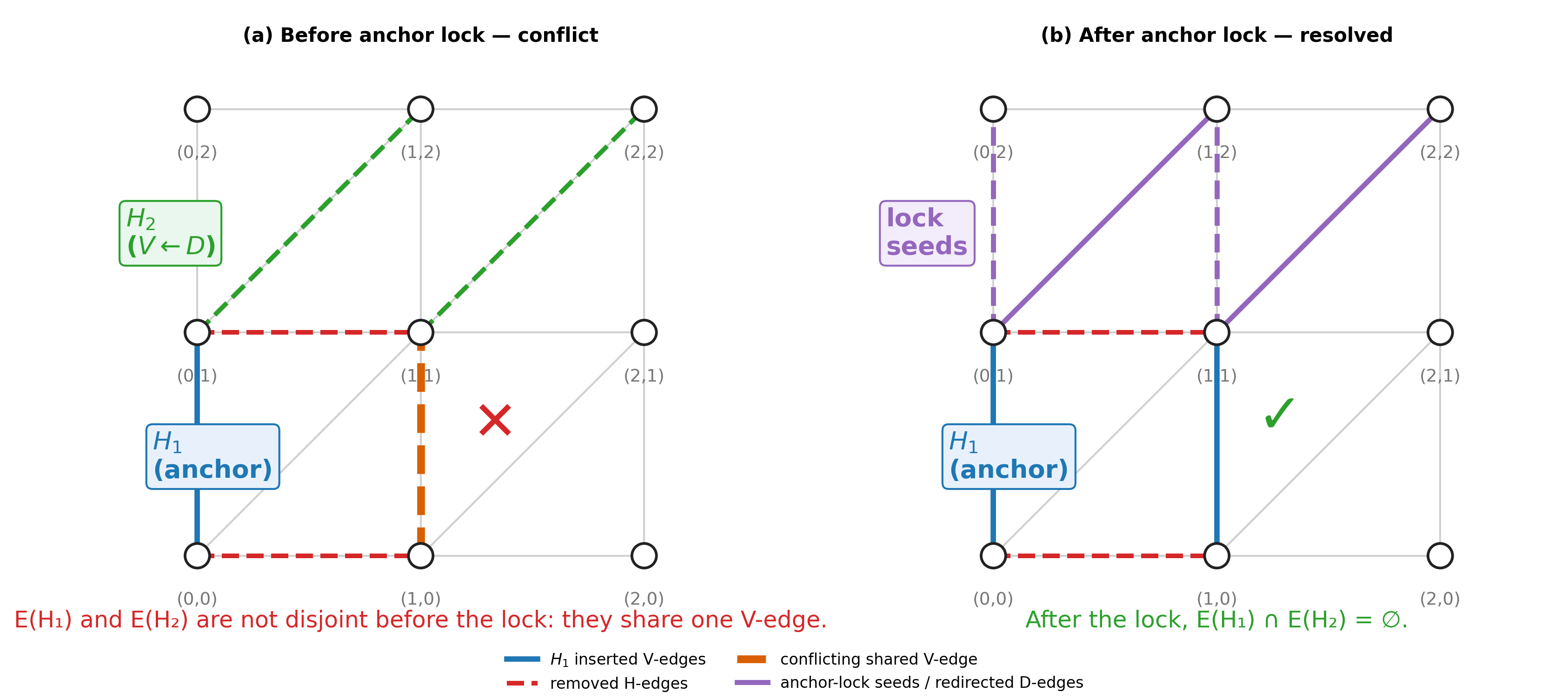}
 \caption{Anchor-lock mechanism. The two designated $V\leftarrow D$ switches delete the vertical bases borrowed by the first $H_1$ switch.}
 \label{fig:anchor}
\end{figure}

\section{Fine-Grained Complement Incidence}

Let $\calD$ be the natural diagonal factor. For a lift $\Phi$, let $\Delta(\Phi)$ be the set of diagonal edges inserted into $H_1$ or $H_2$. These edges are absent from the complement. The connected components of
\begin{equation}
 \calD-\Delta(\Phi)
\end{equation}
are actual diagonal path arcs, together with any uncut diagonal cycles. Denote this component set by $\calP(\Phi)$. Each path arc belongs to one original diagonal cycle and therefore has a coarse label $z\in\ZZ_d$; when a cycle is cut more than once, a second index $\lambda$ is required, and we write the arc as $P_{z,\lambda}$.

The horizontal edges released by the construction of $H_1$ and the vertical edges released by the construction of $H_2$ are precisely the non-diagonal edges of the complement. Contract each component of $\calD-\Delta(\Phi)$ to one vertex and retain these released connector edges.

\begin{Definition}[Fine complement-incidence graph]\label{def:finegraph}
The multigraph $\Gamma(\Phi)$ has vertex set $\calP(\Phi)$. Every released horizontal or vertical edge of the complement gives an edge between the one or two components of $\calD-\Delta(\Phi)$ containing its endpoints.
\end{Definition}

\begin{Definition}[Admissible lift]\label{def:admissible}
A lift $\Phi$ is admissible if:
\begin{enumerate}[label=(A\arabic*),leftmargin=*]
 \item it satisfies anchor lock;
 \item all removed switch bases and inserted diagonal bases are distinct;
 \item $E(H_1)\cap E(H_2)=\emptyset$;
 \item the fine incidence graph $\Gamma(\Phi)$ is connected.
\end{enumerate}
\end{Definition}
Conditions (A1)--(A4) are discharged constructively: the exceptional branches are handled in Section~\ref{sec:exceptional}, while the universal large-$d$ cut orders and explicit incidence words are established in Sections~\ref{sec:universal-lift}--\ref{sec:universal-incidence}.

\begin{Theorem}[Fine-incidence admissible-lift theorem]\label{thm:admissible}
Construct $H_1$ by \eqref{eq:H-skeleton} and $H_2$ from \eqref{eq:V-even} or \eqref{eq:V-odd}. If the chosen lift is admissible, then
\begin{equation}
 H_3=E(\EJ_\alpha)\setminus\bigl(E(H_1)\cup E(H_2)\bigr)
\end{equation}
is a Hamiltonian cycle, and $H_1,H_2,H_3$ are pairwise edge-disjoint.
\end{Theorem}

\begin{proof}
Lemmas~\ref{lem:H1-collapse} and \ref{lem:H2-collapse}, together with (A1)--(A2), make $H_1$ and $H_2$ connected spanning 2-factors. Condition (A3) makes them edge-disjoint. Since $\EJ_\alpha$ is 6-regular, their unused complement $H_3$ is a spanning 2-factor.

Contracting every component of $\calD-\Delta(\Phi)$ in $H_3$ produces exactly $\Gamma(\Phi)$. Contraction preserves the connected-component relation between the contracted graph and $H_3$. Thus (A4) implies that $H_3$ is connected. A connected spanning 2-factor is a Hamiltonian cycle. Pairwise edge-disjointness follows from the definition of the complement.
\end{proof}

\begin{Lemma}[Coarse connector labels]\label{lem:coarse}
Let $z(x,y)=x-y\pmod d$. A released horizontal edge from base $(x,y)$ projects to the coarse pair $\{z,z+1\}$, and a released vertical edge projects to $\{z,z-1\}$. Replacing $(x,y)$ by $(x,y)+d(\pi,\kappa)$ does not change $z$.
\end{Lemma}

\begin{proof}
An $H$ step increases $x-y$ by one, a $V$ step decreases it by one, and a lift changes it by $d(\pi-\kappa)$.
\end{proof}

\begin{Proposition}[Coarse labels are not a connectivity certificate]\label{prop:coarse-insufficient}
If some diagonal cycle is cut into at least two components, connectedness of the graph obtained by projecting every $P_{z,\lambda}$ to $z$ does not imply connectedness of $\Gamma(\Phi)$. Consequently, a word on the $d$ residues alone cannot certify Hamiltonicity of the complement in general.
\end{Proposition}

\begin{proof}
The projection $P_{z,\lambda}\mapsto z$ is noninjective whenever a diagonal cycle has more than one remaining arc. Distinct fine incidence graphs can then have the same coarse projection. For example, with two arcs $P_{0,1},P_{0,2}$ above residue $0$ and two arcs $P_{1,1},P_{1,2}$ above residue $1$, the two disjoint edges $P_{0,1}P_{1,1}$ and $P_{0,2}P_{1,2}$ project to the connected coarse graph consisting of the single edge $\{0,1\}$, while the fine graph is disconnected. Therefore coarse connectedness is insufficient.
\end{proof}

\begin{figure}[H]
 \centering
 \includegraphics[width=0.65\linewidth]{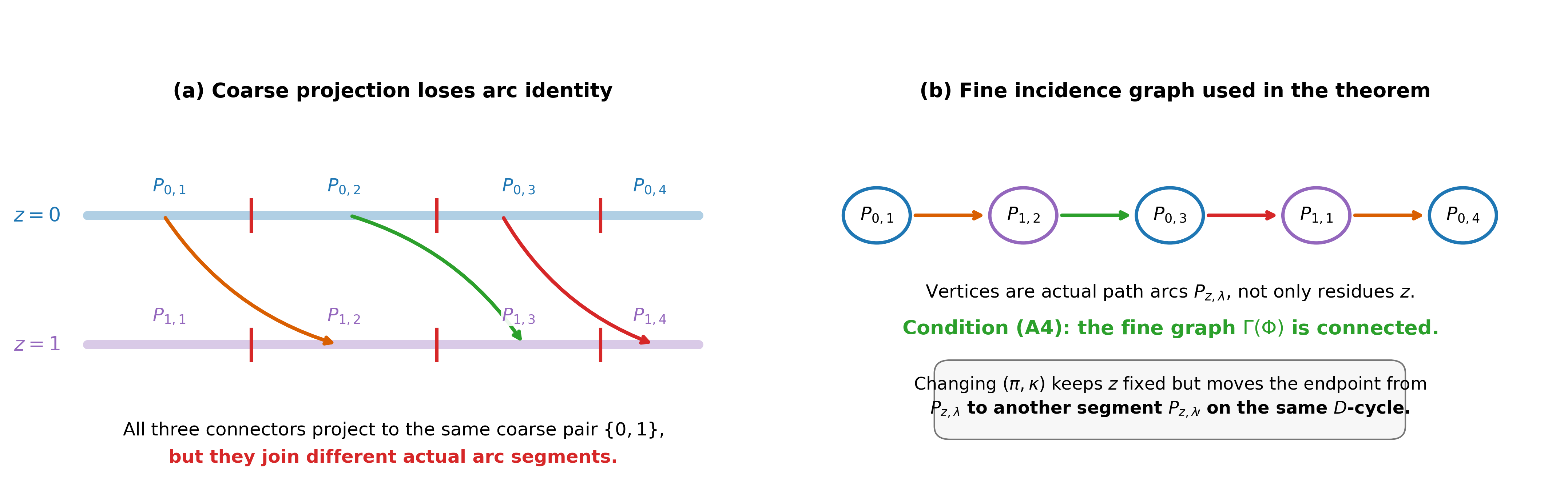}
 \caption{Why complement incidence must be fine-grained. Lift phases preserve the residue $z$ but can move a connector endpoint among different arcs $P_{z,\lambda}$ on the same diagonal cycle.}
 \label{fig:fine-incidence}
\end{figure}

\section{Fine Coordinates and Cut-Order Arithmetic}

Let $a=du$, $b=dv$, $\gcd(u,v)=1$, $Q=u^2+uv+v^2$, and $r=dQ=N/d$.
Because $\gcd(u,Q)=\gcd(u,v^2)=1$, there is a unique
$c\in\{0,1,\ldots,Q-1\}$ satisfying
\begin{equation}
 cu\equiv-v\pmod Q.
\label{eq:cdef}
\end{equation}
For a vertex $(x,y)$ in the $\{1,\tau\}$ basis define
\begin{align}
 z(x,y)&=x-y\pmod d,\label{eq:zcoord}\\
 t(x,y)&=cx+(1-c)y\pmod r.\label{eq:zt}
\end{align}

\begin{Lemma}[Fine diagonal coordinate]\label{lem:fine-coordinate}
On each diagonal component, $t$ is a cyclic coordinate of length $r$. More precisely, $(z,t)(x+1,y+1)=(z,t+1)$, and the map from any $D$-cycle to $\ZZ_r$ is bijective. A lift by $d(\pi,\kappa)$ preserves $z$ and changes $t$ by $d\bigl(c\pi+(1-c)\kappa\bigr)\pmod r$.
\end{Lemma}

\begin{proof}
A $D$ step preserves $x-y$ and increases $cx+(1-c)y$ by $c+(1-c)=1$. Since $D$ has order $r$ by Lemma~\ref{lem:direction-orders}, the $r$ successive values of $t$ are distinct and exhaust $\ZZ_r$. The lift formula follows by substitution in \eqref{eq:zt}.
\end{proof}

\section{Exceptional Constructions for $d=2$ and $d=3$}\label{sec:exceptional}

Because there are $d$ diagonal cycles and each contains $r$ vertices, the map
\begin{equation}
 (x,y)\longmapsto \vtx{z(x,y)}{t(x,y)}
\label{eq:fine-bijection}
\end{equation}
is a bijection from the quotient to $\ZZ_d\times\ZZ_r$. A convenient lattice representative of $\vtx{z}{t}$ is
\begin{equation}
 (x,y)=\bigl(t+(1-c)z,\ t-cz\bigr).
\label{eq:fine-representative}
\end{equation}
The symbols in this paragraph are local to the exceptional branch and are not used in the universal $d\ge4$ construction. For the small-$d$ proofs, let $g\in\{0,\ldots,Q-1\}$ be the unique residue satisfying
\begin{equation}
 gv\equiv u\pmod Q,
\label{eq:gdef}
\end{equation}
and define the coordinate
\begin{equation}
 \vartheta(x,y)=gx+y\pmod r.
\label{eq:vfine}
\end{equation}
On every natural $V$-cycle, $\vartheta$ increases by one at each $V$ step.

\begin{Lemma}[Small-$d$ arithmetic bounds]\label{lem:small-arithmetic}
For $0<u\le v$ and $\gcd(u,v)=1$,
\begin{equation}
 c>u,\qquad v\le g<Q-v,
\label{eq:cg-bounds}
\end{equation}
and
\begin{equation}
 u(g+1)\equiv Q-v\pmod Q.
\label{eq:ug-identity}
\end{equation}
\end{Lemma}

\begin{proof}
If $c\le u$, then $0<cu+v\le u^2+v<Q$, contradicting $Q\mid(cu+v)$. Hence $c>u$.

Suppose $g<v$. Since $Q\mid(gv-u)$ and $|gv-u|<Q$, one has $gv=u$. This forces $v\mid u$ and therefore $u=v=1$. In that case $Q=3$, and the congruence $g\equiv1\pmod3$ with $0\le g<3$ forces $g=1$, contradicting $g<v=1$. Thus $g\ge v$. If $g\ge Q-v$, put $h=Q-g$, so $0<h\le v$. The congruence $gv\equiv u$ gives $Q\mid(hv+u)$, whereas
\begin{equation}
 0<hv+u\le v^2+u<Q,
\end{equation}
again a contradiction. Therefore $g<Q-v$.

Finally,
\begin{equation}
 v\bigl(u(g+1)+v\bigr)\equiv u^2+uv+v^2\equiv0\pmod Q.
\end{equation}
Since $v$ is invertible modulo $Q$, \eqref{eq:ug-identity} follows.
\end{proof}

For a $V\leftarrow D$ switch $B_j$ at seed $s_j$, let $U_j^0$ and $U_j^1$ denote the removed $V$-edge bases $s_j$ and $s_j+D$. Let $B_j^0$ and $B_j^1$ denote the inserted $D$-edge bases $s_j$ and $s_j+V$. The two inserted $D$ connectors in the modified $V$ factor are denoted $\delta_j^0$ and $\delta_j^1$; they join, respectively, the predecessor arcs of $U_j^0,U_j^1$ and the successor arcs of $U_j^0,U_j^1$. The released $V$ connectors in the complement are denoted $b_j^0,b_j^1$. If $X,Y$ are consecutive cuts in a natural $V$-cycle, write $R_p[X>Y]$ for the intervening $V$ path, where $p=x\pmod d$ is the $V$-component label.

\begin{Theorem}[Exceptional $d=2$ construction]\label{thm:d2}
Let $d=2$ and let $0<u\le v$ be coprime. Construct $H_1$ from the horizontal factor by the single $H\leftarrow V$ anchor switch at $(0,0)$. Construct $H_2$ from the vertical factor using three $V\leftarrow D$ switches at
\begin{equation}
 s_0=\vtx{0}{2u},\qquad s_1=\vtx{0}{r-1},\qquad s_2=\vtx{1}{c}.
\label{eq:d2-seeds}
\end{equation}
Equivalently, lattice representatives are $(2u,2u)$, $(r-1,r-1)$, and $(1,0)$. Then $H_1$, $H_2$, and the unused complement $H_3$ are pairwise edge-disjoint Hamiltonian cycles for every reduced ratio $(u,v)$.
\end{Theorem}

\begin{proof}
The anchor joins the two horizontal components. For $H_2$, the removed vertical edges have cyclic orders
\begin{equation}
 \mathcal R_0=(U_1^1,U_2^1,U_0^0),\qquad
 \mathcal R_1=(U_0^1,U_2^0,U_1^0),
\label{eq:d2-V-orders-compact}
\end{equation}
and the six remaining paths occur in the successor order
\begin{equation}
\begin{aligned}
&R_0[U_1^1>U_2^1],\ R_1[U_1^0>U_0^1],\ R_0[U_2^1>U_0^0],\\
&R_1[U_2^0>U_1^0],\ R_0[U_0^0>U_1^1],\ R_1[U_0^1>U_2^0],
\end{aligned}
\label{eq:d2-H2-ranks}
\end{equation}
with outgoing connectors $\delta_1^1,\delta_0^0,\delta_2^1,\delta_1^0,\delta_0^1,\delta_2^0$, respectively. Lemma~\ref{lem:small-arithmetic} makes every gap in \eqref{eq:d2-V-orders-compact} positive, and the local predecessor/successor rule makes the displayed order one cycle.

The diagonal cut orders are
\begin{equation}
 \mathcal C_0=(B_2^1,B_0^0,B_1^0),\qquad
 \mathcal C_1=(B_2^0,B_1^1,B_0^1).
\label{eq:d2-D-orders-compact}
\end{equation}
After anchor lock suppresses $b_1^1,b_2^0$, the six diagonal paths have successor order
\begin{equation}
\begin{aligned}
&P_0[B_1^0>B_2^1],\ P_1[B_0^1>B_2^0],\ P_0[B_0^0>B_1^0],\\
&P_1[B_2^0>B_1^1],\ P_0[B_2^1>B_0^0],\ P_1[B_1^1>B_0^1],
\end{aligned}
\label{eq:d2-H3-ranks}
\end{equation}
with outgoing connectors $a_0^0,b_0^1,b_1^0,b_2^1,b_0^0,a_0^1$. Hence $H_3$ is one cycle. The exact gap lists and endpoint substitutions are collected in Supplementary File S1, Section~S3; they use only Lemma~\ref{lem:small-arithmetic}. Distinct cut positions and anchor lock give pairwise edge-disjointness.
\end{proof}

\begin{Theorem}[Optimal exceptional $d=3$ construction]\label{thm:d3}
Let $d=3$ and let $0<u\le v$ be coprime. Construct $H_1$ from \eqref{eq:H-skeleton}, and construct $H_2$ from the vertical factor using the two $V\leftarrow D$ switches
\begin{equation}
 s_0=(0,-1),\qquad s_1=(-1,-1)
\label{eq:d3-seeds}
\end{equation}
in the $\{1,\tau\}$ basis. Then $H_1$, $H_2$, and the unused complement $H_3$ are pairwise edge-disjoint Hamiltonian cycles for every reduced ratio $(u,v)$, and the total of four switches is minimum.
\end{Theorem}

\begin{proof}
The two $H_1$ switches induce the path $0$--$1$--$2$ on horizontal components. The four removed vertical edges of $H_2$ have orders
\begin{equation}
 \mathcal R_0=(U_0^0,U_1^1),\qquad \mathcal R_1=(U_0^1),\qquad \mathcal R_2=(U_1^0),
\label{eq:d3-V-orders}
\end{equation}
and the four paths are joined cyclically by $\delta_0^1,\delta_0^0,\delta_1^1,\delta_1^0$. Thus $H_2$ is Hamiltonian. The inserted diagonal cuts have orders
\begin{equation}
 \mathcal C_0=(B_1^0,B_0^1,A_1^0),\qquad
 \mathcal C_1=(B_0^0,A_1^1),\qquad \mathcal C_2=(B_1^1).
\label{eq:d3-D-orders}
\end{equation}
After anchor lock removes $b_0^1,b_1^1$, the six complement paths are joined in one successor cycle by $a_0^0,a_1^0,a_0^1,b_1^0,a_1^1,b_0^0$. The fine coordinates, positive gaps, and expanded path order are given in Supplementary File S1, Section~S3; direct substitution shows that all cuts are distinct. Hence the three factors are pairwise edge-disjoint Hamiltonian cycles.

Each natural $H$ and $V$ factor begins with three components, so Lemma~\ref{lem:lowerbound} requires at least two switches for each. The construction uses two plus two and is therefore minimum.
\end{proof}

\begin{Lemma}[Two-switch obstruction when $d=2$]\label{lem:d2-two-switch-obstruction}
Applying exactly two valid local switches to a natural direction factor with two oriented components cannot produce one Hamiltonian cycle: whenever the result is a 2-factor, it has two components.
\end{Lemma}

\begin{proof}
The two cuts on either natural component create the arcs $S>T$ and $T>S$. The predecessor connectors pair the two copies of one arc type, while the successor connectors pair the two copies of the other. The incidence graph is therefore two disjoint 2-cycles. If the switch inventories overlap, the result is not a 2-factor and hence is not Hamiltonian.
\end{proof}

\begin{Theorem}[Optimality of the exceptional $d=2$ branch]\label{thm:d2-optimal}
The four switches in Theorem~\ref{thm:d2} are minimum in the ordered local-switch model.
\end{Theorem}

\begin{proof}
A cost-two decomposition must use one switch in each of the first two factors. Up to direction symmetry, either an $H\leftarrow V$ switch is paired with a $V\leftarrow D$ switch, or an $H\leftarrow D$ switch is paired with a $V\leftarrow D$ switch. In the first case, edge-disjointness would force the two vertical bases inserted by the first switch to equal those removed by the second, implying $H\equiv\pm D$ in the quotient. This would place an element of norm $1$ or $3$ in $(\alpha)$, impossible because $\Norm(\alpha)=4Q\ge12$. In the second case, the released predecessor and successor connectors split the complement into two 2-cycles. Thus cost two is impossible. A cost-three distribution is $1+2$ or $2+1$, and Lemma~\ref{lem:d2-two-switch-obstruction} excludes the factor receiving two switches. Therefore the minimum is four. The lower-budget symmetry classification is summarized in Supplementary File S1, Section~S3.3.
\end{proof}

The corresponding direction-edge counts and all expanded exceptional endpoint certificates are reported in Supplementary File S1, Section~S3.

Index the switches of $H_1$ by $A_0,A_1,\ldots,A_{d-2}$, where $A_0$ is the $H\leftarrow V$ anchor and
\begin{equation}
 A_i:\quad H\leftarrow D\text{ at }(d-1,d-i),\qquad 1\le i\le d-2.
\end{equation}
Index the $V\leftarrow D$ switches of $H_2$ by $B_0,B_1,\ldots,B_{d-2}$. The anchor seeds are
\begin{equation}
 s_{B_0}=(-1,-1),\qquad s_{B_1}=(1,0).
\end{equation}
For $2\le i\le d-2$, define
\begin{equation}
 \varepsilon_i=\begin{cases}
 1,&d\text{ odd and }i=d-2,\\
 0,&\text{otherwise},
 \end{cases}
\label{eq:epsilon}
\end{equation}
and write the canonical lifted seed as
\begin{equation}
 s_{B_i}=(i+dP_i,-\varepsilon_i+dK_i).
\label{eq:canonical-B}
\end{equation}
Thus \eqref{eq:canonical-B} lifts $(i,0)$ except at the unique odd terminal index, where it lifts $(d-2,-1)$ exactly as required by \eqref{eq:V-odd}.
For $i\ge1$, let $A_i^0,A_i^1$ denote the two diagonal edges inserted by $A_i$; let $B_i^0,B_i^1$ denote the corresponding diagonal edges inserted by $B_i$. The anchor $A_0$ inserts no diagonal edge. Lowercase symbols $a_i^s,b_i^s$ denote the released horizontal or vertical connector on side $s$. Under anchor lock, $b_0^1$ and $b_1^0$ are the two vertical edges used by $H_1$ and therefore are not complement connectors.

For cuts $X_1,\ldots,X_q$ on one diagonal cycle, the notation
\begin{equation}
 X_1<_D X_2<_D\cdots<_D X_q<_D X_1
\end{equation}
means that the positive $D$-gaps between successive cuts are all nonzero and sum to $r$.

\begin{Lemma}[Common fine cut coordinates]\label{lem:common-cuts}
For $2\le i\le d-2$, define
\begin{equation}
 \lambda_i\equiv cP_i+(1-c)K_i\pmod Q,
 \qquad e_i\equiv\lambda_i-1\pmod Q,
\label{eq:lambda-e}
\end{equation}
with $0\le\lambda_i,e_i<Q$. For $1\le i\le d-2$, the cuts inserted by $A_i$ have coordinates
\begin{align}
 A_i^0&:\ (z,t)=\bigl(i-1,d-i+c(i-1)\bigr),\label{eq:A0coord}\\
 A_i^1&:\ (z,t)=\bigl(i,d-i+ci\bigr).\label{eq:A1coord}
\end{align}
Whenever $\varepsilon_i=0$, the cuts inserted by $B_i$ have coordinates
\begin{align}
 B_i^0&:\ (z,t)=\bigl(i,ci+d\lambda_i\bigr),\label{eq:B0coord}\\
 B_i^1&:\ (z,t)=\bigl(i-1,ci+d\lambda_i+1-c\bigr).\label{eq:B1coord}
\end{align}
The anchor cuts are
\begin{align}
 B_0^0&=(0,r-1),& B_0^1&=(d-1,-c),\nonumber\\
 B_1^0&=(1,c),& B_1^1&=(0,1).\label{eq:anchor-cuts}
\end{align}
If $d\ge5$ is odd and $k=d-2$, the shifted terminal seed instead gives
\begin{align}
 B_k^0&:\ (z,t)=\bigl(d-1,c(d-1)-1+d\lambda_k\bigr),\label{eq:odd-B0}\\
 B_k^1&:\ (z,t)=\bigl(d-2,c(d-2)+d\lambda_k\bigr).\label{eq:odd-B1}
\end{align}
\end{Lemma}

\begin{proof}
Every formula follows by substituting the corresponding inserted $D$-edge base from Definition~\ref{def:switches} into \eqref{eq:zcoord}--\eqref{eq:zt}. For a regular $B_i$ seed, the lift contributes $d\lambda_i$ to $t$. At the odd terminal seed, the additional offset $(0,-1)$ raises $z$ by one and changes $t$ by $c-1$, giving \eqref{eq:odd-B0}--\eqref{eq:odd-B1}. The anchors follow by direct substitution.
\end{proof}

\begin{Lemma}[Even phase-band reduction]\label{lem:band-reduction}
Assume $d\ge4$ is even. All cuts occur in the cyclic orders
\begin{align}
 C_0&=(B_1^1,A_1^0,B_0^0),\label{eq:C0}\\
 C_1&=(B_1^0,A_2^0,A_1^1,B_2^1),\label{eq:C1}\\
 C_z&=\begin{cases}
 (A_z^1,B_z^0,B_{z+1}^1,A_{z+1}^0),&z\text{ even},\\
 (A_z^1,B_{z+1}^1,B_z^0,A_{z+1}^0),&z\text{ odd},
 \end{cases}\quad 2\le z\le d-3,\label{eq:Cinterior}\\
 C_{d-2}&=(A_{d-2}^1,B_{d-2}^0),\qquad
 C_{d-1}=(B_0^1),\label{eq:Cterminal}
\end{align}
if and only if
\begin{equation}
 \lambda_2\ne0,
 \qquad e_z\le e_{z+1}\ (z\text{ even}),
 \qquad e_{z+1}<e_z\ (z\text{ odd})
\label{eq:zigzag}
\end{equation}
for every $2\le z\le d-3$.
\end{Lemma}

\begin{proof}
For $C_0$, the positive gaps are $d-2$, $r-d$, and $2$. For $C_1$, starting at $B_1^0$, the first two gaps are $d-2$ and $1$; if $\lambda_2\in\{1,\ldots,Q-1\}$, the remaining gaps are $d(\lambda_2-1)+2$ and $d(Q-\lambda_2)-1$, both positive. Conversely, $\lambda_2=0$ destroys the stated order.

For $2\le z\le d-3$, measure positive distance from $A_z^1$. By Lemma~\ref{lem:common-cuts}, the cuts $B_z^0$, $B_{z+1}^1$, and $A_{z+1}^0$ occur at distances
\begin{equation}
 E_z=de_z+z,\qquad F_z=de_{z+1}+z+1,\qquad r-1.
\end{equation}
Both $E_z$ and $F_z$ lie in $\{1,\ldots,r-2\}$. Hence $E_z<F_z$ exactly when $e_z\le e_{z+1}$, whereas $F_z<E_z$ exactly when $e_{z+1}<e_z$. On cycle $d-2$, the two cuts are distinct because their coordinate difference is congruent to $-2\pmod d$, and cycle $d-1$ contains only $B_0^1$. Thus \eqref{eq:zigzag} is necessary and sufficient.
\end{proof}

\begin{Lemma}[Odd phase-band reduction]\label{lem:odd-band-reduction}
Assume $d\ge5$ is odd. The cuts occur in the orders $C_0$ and $C_1$ of \eqref{eq:C0}--\eqref{eq:C1}, in the orders \eqref{eq:Cinterior} for $2\le z\le d-4$, and in the terminal orders
\begin{align}
 C_{d-3}^{\rm odd}&=(A_{d-3}^1,B_{d-3}^0,A_{d-2}^0),\label{eq:odd-C1}\\
 C_{d-2}^{\rm odd}&=(A_{d-2}^1,B_{d-2}^1),\qquad
 C_{d-1}^{\rm odd}=(B_{d-2}^0,B_0^1).\label{eq:odd-C2}
\end{align}
These orders hold if and only if
\begin{equation}
 \lambda_2\ne0,
 \qquad e_z\le e_{z+1}\ (z\text{ even}),
 \qquad e_{z+1}<e_z\ (z\text{ odd})
\label{eq:odd-zigzag}
\end{equation}
for every $2\le z\le d-4$.
\end{Lemma}

\begin{proof}
The calculations for $C_0,C_1$ and cycles $2,\ldots,d-4$ are identical to the even case. We derive all three terminal orders explicitly. Put $k=d-2$.

On cycle $d-3$, the three coordinates, written from $A_{d-3}^1$, are
\begin{align}
 t(A_{d-3}^1)&=c(d-3)+3,\nonumber\\
 t(B_{d-3}^0)&=c(d-3)+d\lambda_{d-3},\nonumber\\
 t(A_{d-2}^0)&=c(d-3)+2.
\end{align}
The hypotheses imply $\lambda_{d-3}\ne0$: for $d=5$ this is $\lambda_2\ne0$, while for $d\ge7$ the last odd comparison in \eqref{eq:odd-zigzag} rules out $e_{d-3}=Q-1$, equivalently $\lambda_{d-3}=0$. Hence the three successive positive gaps in \eqref{eq:odd-C1} are
\begin{equation}
 d\lambda_{d-3}-3,\qquad d(Q-\lambda_{d-3})+2,\qquad 1.
\label{eq:odd-terminal-gaps-1}
\end{equation}
They are positive and sum to $dQ=r$.

On cycle $d-2$, set $\eta=\lambda_k$. The successive positive gaps for
$A_k^1<_D B_k^1<_D A_k^1$ are
\begin{equation}
 (g_{d-2}^{0},g_{d-2}^{1})=
 \begin{cases}
 (r-2,2),&\eta=0,\\
 (d\eta-2,d(Q-\eta)+2),&1\le\eta<Q.
 \end{cases}
\label{eq:odd-terminal-gaps-2}
\end{equation}
Both entries are positive and their sum is $r$.

Finally, put $\mu\equiv c+\lambda_k\pmod Q$, $0\le\mu<Q$. From \eqref{eq:odd-B0} and \eqref{eq:anchor-cuts}, the successive positive gaps for
$B_k^0<_D B_0^1<_D B_k^0$ are
\begin{equation}
 (g_{d-1}^{0},g_{d-1}^{1})=
 \begin{cases}
 (1,r-1),&\mu=0,\\
 (d(Q-\mu)+1,d\mu-1),&1\le\mu<Q.
 \end{cases}
\label{eq:odd-terminal-gaps-3}
\end{equation}
Again both entries are positive and sum to $r$. Thus the three terminal orders are proved without an enumeration or an implicit choice of representatives. The only remaining conditions are exactly \eqref{eq:odd-zigzag}.
\end{proof}

\section{Universal Alternating Lift}\label{sec:universal-lift}

Assume from now on that $d\ge4$. For every non-anchor switch index $2\le i\le d-2$, define
\begin{equation}
 \theta_i=
 \begin{cases}
  1,&i\text{ even},\\
  2,&i\text{ odd},
 \end{cases}
 \qquad P_i=K_i=\theta_i.
\label{eq:universal-phase}
\end{equation}
Together with \eqref{eq:epsilon}, the resulting seed is
\begin{equation}
 s_{B_i}=\bigl(i+d\theta_i,-\varepsilon_i+d\theta_i\bigr).
\label{eq:universal-seed}
\end{equation}
Thus the lift translates each reduced seed by either $dD$ or $2dD$ in the diagonal direction. The choice depends only on the parity of the switch index and not on the reduced ratio $(u,v)$.

\begin{Lemma}[Ratio-independent phase bands]\label{lem:universal-bands}
For every coprime $0<u\le v$ and every $d\ge4$, the lift \eqref{eq:universal-phase} satisfies
\begin{equation}
 \lambda_i=\theta_i,
 \qquad
 e_i=
 \begin{cases}
 0,&i\text{ even},\\
 1,&i\text{ odd}.
 \end{cases}
\label{eq:universal-bands}
\end{equation}
In particular, $\lambda_2=1\ne0$.
\end{Lemma}

\begin{proof}
Because $u,v$ are positive, $Q=u^2+uv+v^2\ge3$. Substituting $P_i=K_i=\theta_i$ into \eqref{eq:lambda-e} gives
\begin{equation}
 \lambda_i\equiv c\theta_i+(1-c)\theta_i
 =\theta_i\pmod Q.
\end{equation}
The integers $1$ and $2$ lie in $\{0,1,\ldots,Q-1\}$, so this congruence already gives the exact representatives. They are distinct and nonzero even at the extremal value $Q=3$, which occurs for the coprime positive pair $u=v=1$; see Remark~\ref{rem:symmetric}. Hence $e_i=\lambda_i-1$ is zero at even indices and one at odd indices, and $\lambda_2=1$. No avoidance of any additional residue is required: Lemmas~\ref{lem:band-reduction} and \ref{lem:odd-band-reduction} list all phase-dependent cut-order conditions, and these are audited explicitly in Lemma~\ref{lem:certified-orders}.
\end{proof}

\begin{Corollary}[Ratio-independent dependency closure]\label{cor:no-ratio-cases}
For every $d\ge4$, the phase choice \eqref{eq:universal-phase} is valid for every coprime normalized pair $0<u\le v$ without evaluating $c$ or dividing the reduced ratios into arithmetic cases. After the common coordinate formulas are established, the remainder of the construction depends on $(u,v)$ only through the fixed values $\lambda_i\in\{1,2\}$ and $e_i\in\{0,1\}$.
\end{Corollary}

\begin{proof}
Lemma~\ref{lem:common-cuts} isolates all reduced-ratio dependence in the fine-coordinate constant $c$ and in the combinations $\lambda_i=cP_i+(1-c)K_i$. Lemma~\ref{lem:universal-bands} cancels $c$ and fixes those combinations exactly. Lemmas~\ref{lem:band-reduction} and \ref{lem:odd-band-reduction} then reduce every cut order, including the boundary cycles, to conditions checked in Lemma~\ref{lem:certified-orders}. Finally, Lemmas~\ref{lem:fine-block} and \ref{lem:odd-fine-block} use only those cyclic cut orders and the local connector identities. Thus no later step receives any unreduced arithmetic information about $(u,v)$.
\end{proof}

\begin{Remark}[Symmetric reduced ratio]\label{rem:symmetric}
The case $u=v$ is not degenerate within the normalized positive domain: coprimality forces $u=v=1$ and $Q=3$. The residues $1$ and $2$ are then distinct and nonzero, so the universal lift and both parity proofs apply unchanged. Supplementary Figure~S2 illustrates this symmetric ratio at $d=4$.
\end{Remark}

\begin{Lemma}[Universal cut orders for both parities]\label{lem:certified-orders}
For every coprime $0<u\le v$, the universal lift \eqref{eq:universal-phase} satisfies \eqref{eq:zigzag} when $d$ is even and \eqref{eq:odd-zigzag} when $d$ is odd. Consequently all diagonal cuts are distinct and occur in the cyclic orders of Lemma~\ref{lem:band-reduction} or Lemma~\ref{lem:odd-band-reduction}, respectively.
\end{Lemma}

\begin{proof}
For every generic interior pair, Lemma~\ref{lem:universal-bands} gives $(e_z,e_{z+1})=(0,1)$ when $z$ is even and $(1,0)$ when $z$ is odd. Hence the weak and strict inequalities in \eqref{eq:zigzag} and \eqref{eq:odd-zigzag} hold throughout their stated ranges.

It remains to check that no boundary condition is hidden outside those comparisons. On $C_0$, the gaps $d-2$, $r-d=d(Q-1)$, and $2$ are positive independently of every phase. On $C_1$, the only phase requirement is $\lambda_2\ne0$; here $\lambda_2=1$, and the last two gaps in Lemma~\ref{lem:band-reduction} become $2$ and $d(Q-1)-1>0$. For even $d$, cycle $C_{d-2}$ contains two cuts whose coordinates differ by $-2\pmod d$, which is nonzero for $d\ge4$, and $C_{d-1}$ contains one cut. For odd $d$, $\lambda_{d-3}\in\{1,2\}$, so the first terminal gap $d\lambda_{d-3}-3$ in \eqref{eq:odd-terminal-gaps-1} is positive for $d\ge5$; the other two terminal gap pairs \eqref{eq:odd-terminal-gaps-2}--\eqref{eq:odd-terminal-gaps-3} were proved positive for every possible terminal residue. Therefore the universal lift satisfies every interior and boundary hypothesis of Lemmas~\ref{lem:band-reduction} and \ref{lem:odd-band-reduction}, including the tight case $Q=3$. The stated cyclic orders and cut distinctness follow.
\end{proof}

\begin{Example}[Even base trace for $d=6$]\label{ex:d6-trace}
Take $(u,v)=(2,5)$ and $d=6$. Then $Q=39$, $N=1404$, and the three non-anchor phases are $(1,1),(2,2),(1,1)$, giving seeds $(8,6),(15,12),(10,6)$ and phase bands $(e_2,e_3,e_4)=(0,1,0)$ without evaluating $c$. The even successor certificate has $m=4d-6=18$ states. Because the interior range $4\le k\le d-4$ is empty at $d=6$, the rank word consists of four transparent pieces:
\begin{center}
\begin{tabular}{@{}llll@{}}
\toprule
piece & parameter range & number of entries & ranks\\
\midrule
initial & -- & $3$ & $0,1,2$\\
terminal & $k=d-2=4$ & $6$ & $3,\ldots,8$\\
descent & $z=3,2$ & $2$ & $9,10$\\
anchor closure & -- & $7$ & $11,\ldots,17$\\
\bottomrule
\end{tabular}
\end{center}
The last target of each piece is the first source of the next, and the final anchor connector returns rank $17$ to rank $0$. Thus this smallest even case containing the full terminal and descent mechanisms already displays the block-chain invariant used for every larger even $d$. For the same ratio with $d=5$, the terminal seed is shifted to $s_{B_3}=(13,9)$, and the odd certificate replaces the descent by the shifted seam and rise blocks. Supplementary Figure~S2 gives the complete three-cycle drawing for the smallest symmetric instance in the universal branch, $u=v=1$, $d=4$, $N=48$.
\end{Example}

\begin{Lemma}[Validity and pairwise disjointness of the first two factors]\label{lem:first-two-valid}
For every coprime $0<u\le v$, every $d\ge4$, and the universal lift \eqref{eq:universal-phase}, all switch removals are valid, $H_1$ and $H_2$ are Hamiltonian cycles, and $E(H_1)\cap E(H_2)=\emptyset$.
\end{Lemma}

\begin{proof}
For an even $d$, the two removed vertical bases of $B_i$ have diagonal label $i\pmod d$. For an odd $d$, this remains true for $0\le i\le d-3$, while the shifted terminal switch $B_{d-2}$ has label $d-1$. Thus distinct switch indices occupy distinct diagonal labels, and the two bases within one switch differ by the nonzero step $D$. The analogous statement for $A_i$, $i\ge1$, follows from their distinct diagonal labels; the two anchor removals of $A_0$ are separated from the remaining bases by their explicit $(z,t)$ coordinates. Hence every switch removal is valid. Since the component-incidence edges of Lemmas~\ref{lem:H1-collapse} and \ref{lem:H2-collapse} form trees, the first two factors are Hamiltonian.

The only vertical edges inserted into $H_1$ are the two anchor edges, and Lemma~\ref{lem:anchor} removes them from $H_2$. The second factor contains no horizontal insertion. Finally, Lemma~\ref{lem:certified-orders} places every $A$-cut and $B$-cut at a distinct position on its diagonal cycle in either parity, so no inserted diagonal edge is shared. Hence $H_1$ and $H_2$ are edge-disjoint.
\end{proof}

\section{Universal Fine-Incidence Cycle}\label{sec:universal-incidence}

If $X$ and $Y$ are consecutive cuts in positive $D$-order on cycle $z$, write
\begin{equation}
 P_z[X>Y]
\end{equation}
for the remaining diagonal path beginning immediately after cut $X$ and ending immediately before cut $Y$.

\subsection{Proof map and block-chain invariant}

The complement proof separates one repeating interior calculation from a constant number of boundary checks. Table~\ref{tab:proof-map} records the dependency structure before the detailed successor words are introduced.

\begin{table}[H]
\centering
\caption{Proof map for the universal complement. Only the six-entry interior block grows with $d$; every seam and closure block has constant size.}
\label{tab:proof-map}
\small
\begin{adjustbox}{max width=\textwidth}
\begin{tabularx}{1.3\textwidth}{@{}p{0.21\textwidth}XX@{}}
\toprule
Obligation & Structural mechanism & Finite boundary work\\
\midrule
Arc inventory & distinct cut orders from Lemmas~\ref{lem:band-reduction}--\ref{lem:certified-orders} & $d=4$ deduplication and the shifted odd terminal cycles\\
Local successor step & predecessor/successor lookup on one cut cycle & anchor identities and parity-specific seams\\
Unbounded transitions & one six-entry block $\mathsf M_k$ shared by both parities & none beyond its two endpoint cut lists\\
Rank coverage & consecutive intervals in Table~\ref{tab:rank-intervals} & initial, terminal, and closure intervals\\
Source coverage & each interior arc is assigned to one adjacent block & four boundary cycle inventories\\
Connector coverage & each released side appears once & omission of the two anchor-locked vertical sides\\
\bottomrule
\end{tabularx}
\end{adjustbox}
\end{table}

\begin{Lemma}[Connector endpoint identities]\label{lem:connector-identities}
Let a diagonal cut cycle have cyclic order $C_z=(X_0,X_1,\ldots,X_{q-1})$, with subscripts read modulo $q$. Then
\begin{align}
 \operatorname{pred}(X_j)&=\Arc{z}{X_{j-1}}{X_j},\nonumber\\
 \operatorname{succ}(X_j)&=\Arc{z}{X_j}{X_{j+1}}.
\label{eq:pred-succ-lookup}
\end{align}
For every $i\ge1$, connector $a_i^0$ joins the predecessor arcs of $A_i^0$ and $A_i^1$, while $a_i^1$ joins their successor arcs. For every $i\ge0$, connector $b_i^0$ joins the predecessor arcs of $B_i^0$ and $B_i^1$, while $b_i^1$ joins their successor arcs, except that $b_0^1$ and $b_1^0$ are suppressed by anchor lock. The two anchor connectors satisfy
\begin{align}
 a_0^0&:\ \operatorname{pred}(B_1^0)\longleftrightarrow\operatorname{pred}(B_1^1),\label{eq:anchor-endpoint-0}\\
 a_0^1&:\ \operatorname{pred}(B_1^1)\longleftrightarrow\operatorname{succ}(B_0^1).
\label{eq:anchor-endpoint-1}
\end{align}
These statements hold for both parity classes, including the shifted odd terminal cuts.
\end{Lemma}

\begin{proof}
Deleting a cut edge turns its two incident vertices into the terminal vertices of the predecessor and successor diagonal arcs in \eqref{eq:pred-succ-lookup}. For an $H\leftarrow D$ switch, side $0$ removes the horizontal edge joining the two bases of its inserted $D$-edges. Those two vertices are the terminal vertices immediately before the two cuts, so the released horizontal edge joins the two predecessor arcs. Side $1$ is the translate by $D$; its endpoints are the initial vertices immediately after the cuts, so it joins the two successor arcs. The same argument applies verbatim to a $V\leftarrow D$ switch, with the released connector vertical rather than horizontal. It depends only on the local switch parallelogram and therefore is unchanged by the odd terminal shift.

For the anchor $A_0$, the two released horizontal edges are $a_0^0=[(0,0),(1,0)]$ and $a_0^1=[(0,1),(1,1)]$. Their endpoint locations can be read without a drawing. Substitution in \eqref{eq:zcoord}--\eqref{eq:zt} gives
\begin{center}
\begin{tabular}{@{}ccc@{}}
\toprule
vertex & $(z,t)$ & containing diagonal arc\\
\midrule
$(0,0)$ & $(0,0)$ & $\operatorname{pred}(B_1^1)=\Arc{0}{B_0^0}{B_1^1}$\\
$(1,0)$ & $(1,c)$ & $\operatorname{pred}(B_1^0)=\Arc{1}{B_2^1}{B_1^0}$\\
$(0,1)$ & $(d-1,1-c)$ & $\operatorname{succ}(B_0^1)=\Arc{d-1}{B_0^1}{B_0^1}$\\
$(1,1)$ & $(0,1)$ & $\operatorname{pred}(B_1^1)=\Arc{0}{B_0^0}{B_1^1}$
\end{tabular}
\end{center}
where the arc identifications use $C_0,C_1,C_{d-1}$ and the anchor-cut coordinates \eqref{eq:anchor-cuts}. Thus $a_0^0$ joins $\operatorname{pred}(B_1^0)$ to $\operatorname{pred}(B_1^1)$, while $a_0^1$ joins $\operatorname{pred}(B_1^1)$ to $\operatorname{succ}(B_0^1)$, proving \eqref{eq:anchor-endpoint-0}--\eqref{eq:anchor-endpoint-1}. Under anchor lock, the other two vertical sides are exactly $b_0^1$ and $b_1^0$, already used by $H_1$, which explains their suppression from the complement.
\end{proof}

For $d\ge4$, the cardinality used by the rank certificate is forced by the switch inventory. The $d-2$ diagonal-inserting switches of $H_1$ contribute $2(d-2)$ distinct diagonal cuts, and the $d-1$ switches of $H_2$ contribute $2(d-1)$ more. Lemma~\ref{lem:certified-orders} makes all of them distinct and places at least one cut on every diagonal cycle, so the remaining diagonal-path inventory has
\begin{equation}
 2(d-2)+2(d-1)=4d-6
\label{eq:arc-count}
\end{equation}
actual arcs. Likewise, the two first-factor connector sides and two second-factor connector sides give $2(d-1)+2(d-1)$ candidates, and anchor lock suppresses exactly $b_0^1$ and $b_1^0$. Hence the available connector inventory also has $4d-6$ members.

\begin{Definition}[Successor permutation and rank certificate]\label{def:rank-certificate}
For either parity, let $\mathcal P_d=\dot\bigcup_z\mathcal A_z$ be the set of remaining diagonal arcs and let $m=4d-6$. A \emph{rank certificate} is a bijection
\begin{equation}
 R_d:\mathcal P_d\longrightarrow\ZZ_m
\label{eq:rank-map}
\end{equation}
together with a connector map $e_d:\mathcal P_d\to\mathcal E_d$ such that the connector $e_d(P)$ joins $P$ to an arc $\sigma_d(P)$ satisfying
\begin{equation}
 R_d(\sigma_d(P))\equiv R_d(P)+1\pmod m.
\label{eq:rank-increment}
\end{equation}
If $e_d$ is also bijective, then $\Gamma(\Phi)$ is the single cycle with cyclic order
\begin{equation}
 R_d^{-1}(0),R_d^{-1}(1),\ldots,R_d^{-1}(m-1).
\label{eq:rank-cycle-order}
\end{equation}
\end{Definition}

\begin{Lemma}[Rank-certificate criterion]\label{lem:rank-criterion}
The conditions in Definition~\ref{def:rank-certificate} imply that the fine incidence graph is one cycle.
\end{Lemma}

\begin{proof}
Equation~\eqref{eq:rank-increment} conjugates the successor permutation $\sigma_d$ to addition of one on $\ZZ_m$. Addition of one has a single orbit containing all $m$ residues. The bijection $e_d$ uses every available incidence connector exactly once, so the resulting orbit is precisely the complete fine incidence graph.
\end{proof}

\begin{Definition}[Block-chain invariant]\label{def:block-invariant}
A successor block is an ordered list
\[
\mathsf B=(\StepPair{P_0}{e_0},\ldots,\StepPair{P_{q-1}}{e_{q-1}}).
\]
It satisfies the block-chain invariant when: (i) the connector $e_j$ joins $P_j$ to $P_{j+1}$ for $0\le j<q-1$; (ii) its source arcs and connectors are pairwise distinct; and (iii) the assigned ranks are a consecutive interval. Two blocks chain when the last connector of the first joins its last source to the first source of the second.
\end{Definition}

\begin{Lemma}[Concatenation of successor blocks]\label{lem:block-concatenation}
Suppose a cyclic list of successor blocks satisfies the block-chain invariant, consecutive blocks chain, the source inventories partition $\mathcal P_d$, and the connector inventories partition $\mathcal E_d$. Then their concatenation is a rank certificate in the sense of Definition~\ref{def:rank-certificate}.
\end{Lemma}

\begin{proof}
The within-block and between-block conditions make every connector advance to the next source in the concatenated order. Consecutive rank intervals therefore define a bijection from the source inventory to $\ZZ_{|\mathcal P_d|}$ with increment one modulo its size. The two inventory partitions make the source and connector maps bijective, so Definition~\ref{def:rank-certificate} holds.
\end{proof}

\begin{table}[H]
\centering
\caption{Block sizes and rank intervals for the parity-complete successor certificates. Full source--connector--target dictionaries are in Supplementary File S1, Sections~S5--S6. Empty index ranges contribute no ranks.}
\label{tab:rank-intervals}
\small
\begin{adjustbox}{max width=\textwidth}
\begin{tabularx}{1.2\textwidth}{@{}c l l c l@{}}
\toprule
Parity & Block & Parameter range & Size & Rank interval in $\ZZ_{4d-6}$\\
\midrule
Even & Initial & -- & 3 & $0,1,2$\\
& Interior $M_k$ & even $4\le k\le d-4$ & 6 & $3k-9,\ldots,3k-4$\\
& Terminal & -- & 6 & $3d-15,\ldots,3d-10$\\
& Descent $R_z$ & $d-3\ge z\ge2$ & 1 & $3d-9+(d-3-z)$\\
& Anchor closure & -- & 7 & $4d-13,\ldots,4d-7$\\
\addlinespace
Odd & Initial & -- & 3 & $0,1,2$\\
& Interior $M_k$ & even $4\le k\le d-3$ & 6 & $3k-9,\ldots,3k-4$\\
& Shifted seam & -- & 3 & $3d-12,\ldots,3d-10$\\
& Rise $R_z$ & $1\le z\le d-4$ & 1 & $3d+z-10$\\
& Terminal & -- & 3 & $4d-13,\ldots,4d-11$\\
& Closure & -- & 4 & $4d-10,\ldots,4d-7$\\
\bottomrule
\end{tabularx}
\end{adjustbox}
\end{table}

\begin{Lemma}[Exhaustiveness of the interior transition types]\label{lem:interior-exhaustive}
Under the universal alternating lift, every transition whose indices remain outside the fixed boundary cycles is, up to translation of the indices, either the six-entry block $\mathsf M_k$ in \eqref{eq:generic-rank-block} or one of the one-entry rise/descent transitions in Table~\ref{tab:rank-intervals}. No additional transition type appears as $d$ increases.
\end{Lemma}

\begin{proof}
Lemma~\ref{lem:universal-bands} gives the alternating phase pairs $(e_z,e_{z+1})=(0,1)$ for even $z$ and $(1,0)$ for odd $z$. Lemmas~\ref{lem:band-reduction} and \ref{lem:odd-band-reduction} therefore leave only the two alternating four-cut orders on every nonboundary diagonal cycle. By Lemma~\ref{lem:connector-identities}, a connector at index $i$ can meet only predecessor or successor arcs at the two cuts carrying that same switch index. Consequently, for an even interior index $k$, the complete local neighborhood is determined by the four cut lists $C_{k-2},C_{k-1},C_k,C_{k+1}$, and substitution of their two alternating orders yields exactly the six source--connector entries of $\mathsf M_k$. The only unused connector side on a generic cut cycle produces the single rise/descent transition. Deviations from these two translated patterns can occur only on $C_0,C_1$ and the parity-specific terminal cycles, all of which belong to the constant-size initial, seam, terminal, or closure blocks. Hence the listed families are exhaustive.
\end{proof}

\begin{Corollary}[Parity-template extension by two]\label{cor:template-extension}
Fix one parity. Passing from the block template at $d$ to the template at $d+2$ introduces one translated copy of $\mathsf M_k$ and two additional one-entry rise/descent transitions; the constant-size terminal template is translated to the new boundary, and no new endpoint identity type is introduced. The number of states increases by
\[
6+1+1=8=(4(d+2)-6)-(4d-6).
\]
Thus, after the base templates and the constant boundary blocks have been verified, the block-chain and inventory arguments extend inductively to every larger $d$ of the same parity.
\end{Corollary}

\begin{proof}
For even $d\ge6$, the template order is initial, $\mathsf M_4,\mathsf M_6,\ldots,\mathsf M_{d-4}$, terminal, descent, and closure. Replacing $d$ by $d+2$ adds $\mathsf M_{d-2}$ and the two descent indices $d-1,d-2$, while the terminal formulas are the same symbolic formulas with their boundary index shifted by two. For odd $d\ge5$, replacing $d$ by $d+2$ adds $\mathsf M_{d-1}$ and the two rise indices $d-3,d-2$, while the shifted seam and terminal blocks retain their symbolic forms at the new boundary. Lemma~\ref{lem:interior-exhaustive} supplies the new interior and one-entry endpoint identities, and the constant boundary identities are independent of the number of preceding interior blocks. The added blocks use exactly the eight new arc and connector entries, so both inventory partitions and the block-chain invariant are preserved.
\end{proof}

The only unbounded transition family is the following six-entry interior block, shared by both parities:
\begin{equation}
\begin{aligned}
\mathsf M_k=(&\StepPair{\Arc{k-1}{A_{k-1}^1}{B_k^1}}{b_k^0},
\StepPair{\Arc{k}{A_k^1}{B_k^0}}{a_k^1},
\StepPair{\Arc{k-1}{A_k^0}{A_{k-1}^1}}{a_{k-1}^0},\\
&\StepPair{\Arc{k-2}{B_{k-1}^1}{A_{k-1}^0}}{b_{k-1}^1},
\StepPair{\Arc{k-1}{B_{k-1}^0}{A_k^0}}{a_k^0},
\StepPair{\Arc{k}{A_{k+1}^0}{A_k^1}}{a_{k+1}^1}).
\end{aligned}
\label{eq:generic-rank-block}
\end{equation}
All remaining blocks have constant size. Their explicit entries, including the deduplicated $d=4$ word and the shifted $d=5$ boundary, appear in Supplementary File S1; the main proof below uses only their stated ranks, endpoint types, and coverage sets.

\begin{Lemma}[Even universal fine-incidence block]\label{lem:fine-block}
Assume $d$ is even and the cut orders \eqref{eq:C0}--\eqref{eq:Cterminal}. Then $\Gamma(\Phi)$ admits a rank certificate on all $4d-6$ diagonal arcs and is one cycle.
\end{Lemma}

\begin{proof}
The exceptional even base $d=4$ is the deduplicated ten-entry word in Supplementary Section~S5. For $d=6$, Example~\ref{ex:d6-trace} gives the base block order: initial, terminal, two descent entries, and anchor closure. The endpoint identities follow from Lemma~\ref{lem:connector-identities}, the ranks cover $\ZZ_{18}$, and the source and connector inventories are complete.

Assume now that the block-chain and the two inventory partitions hold for an even value $d\ge6$. Corollary~\ref{cor:template-extension} describes the template at $d+2$: it adds the translated interior block $\mathsf M_{d-2}$ and two new descent transitions, moves the constant terminal template to the new boundary, and introduces no new endpoint identity type. The new six-entry block and the two singleton transitions account for exactly the eight additional arcs and connectors. All earlier initial and interior identities remain unchanged, and the translated terminal and closure identities are the same constant substitutions in the new terminal cut lists. Hence the block-chain invariant and both inventory partitions are preserved. This proves the claim for every even $d\ge6$ by induction.

For completeness, the generic six-step identity used in the inductive extension follows directly from Lemma~\ref{lem:connector-identities}: the six connectors in \eqref{eq:generic-rank-block} join, in order, predecessor pairs at $B_k$, successor pairs at $A_k$, predecessor pairs at $A_{k-1}$, successor pairs at $B_{k-1}$, predecessor pairs at $A_k$, and successor pairs at $A_{k+1}$. The induction maintains bijective source and connector inventories onto $\mathcal P_d$ and
\begin{equation}
 \mathcal E_d=\{a_i^s:0\le i\le d-2,\ s\in\{0,1\}\}\ \dot\cup\
 \bigl(\{b_i^s:0\le i\le d-2,\ s\in\{0,1\}\}\setminus\{b_0^1,b_1^0\}\bigr),
\label{eq:connector-set-compact}
\end{equation}
respectively. Lemma~\ref{lem:block-concatenation} yields the rank certificate, and Lemma~\ref{lem:rank-criterion} gives one cycle.
\end{proof}

\begin{Lemma}[Odd universal fine-incidence block]\label{lem:odd-fine-block}
Assume $d\ge5$ is odd and the cut orders of Lemma~\ref{lem:odd-band-reduction}. Then $\Gamma(\Phi)$ admits a rank certificate on all $4d-6$ diagonal arcs and is one cycle.
\end{Lemma}

\begin{proof}
For $d=5$, the interior range is empty. The fourteen-entry base certificate in Supplementary Section~S6 consists of the initial block, shifted seam, one rise entry, terminal block, and closure; the displayed endpoint identities give one closed rank cycle and use every source arc and connector exactly once.

Assume that the block-chain and both inventory partitions hold for an odd value $d\ge5$. By Corollary~\ref{cor:template-extension}, the template at $d+2$ adds the translated block $\mathsf M_{d-1}$ and two new rise transitions, and translates the same constant-size shifted seam and terminal templates to the new boundary. Lemma~\ref{lem:interior-exhaustive} shows that these are the only new transition types. The added six-entry block and the two singleton transitions account for the eight new arcs and connectors, while all earlier initial and interior identities remain unchanged. The boundary substitutions retain their symbolic forms after the index shift. Hence the block-chain invariant and both inventory partitions are preserved, proving the claim for every odd $d$ by induction.

The generic interior endpoint identity is the same as in the even proof. Each rise transition uses the complementary side of the relevant $B$ connector and enters the next source arc on the adjacent cut cycle; the final rise enters the shifted terminal cycle. The constant seam, terminal, and closure substitutions use only $C_0,C_1,C_2,C_{d-3},C_{d-2},C_{d-1}$ and the anchor identities. Their expanded forms are recorded in Supplementary Section~S6. Thus Lemma~\ref{lem:block-concatenation} gives the rank certificate, and Lemma~\ref{lem:rank-criterion} proves that $\Gamma(\Phi)$ is one cycle.
\end{proof}

\section{Universal Decomposition}

\begin{Theorem}[Complete six-regular non-axis local EJ decomposition]\label{thm:main}
Let $\alpha\in\ZZ[\rho]\setminus\{0\}$ be non-axis in the sense of Definition~\ref{def:axis}, and suppose that $\EJ_\alpha$ is six-regular simple. Equivalently, after the normalization of Lemma~\ref{lem:orbit-normalization}, choose $a+b\rho$ with $0<a\le b$ and $(a,b)\ne(1,1)$. Write $a=du$, $b=dv$ with $\gcd(u,v)=1$. Then $\EJ_\alpha$ decomposes into three pairwise edge-disjoint Hamiltonian cycles for every admitted $d\ge1$. More precisely:
\begin{enumerate}[label=\roman*),leftmargin=*]
 \item If $d=1$, the three natural unit-direction factors are Hamiltonian cycles.
 \item If $d=2$ or $d=3$, use the exceptional formulas of Theorems~\ref{thm:d2} and \ref{thm:d3}.
 \item If $d\ge4$, construct $H_1$ by \eqref{eq:H-skeleton}; construct $H_2$ from the parity-appropriate reduced seeds \eqref{eq:V-even} or \eqref{eq:V-odd}, anchor lock, and the universal phases \eqref{eq:universal-phase}; and let $H_3$ be the unused complement.
\end{enumerate}
The conclusion has no parity or reduced-ratio restriction inside this exact domain. The fixed branches $d=2$ and $d=3$ require no even/odd case split; the parity-dependent terminal geometry arises only in the variable branch $d\ge4$.
\end{Theorem}

\begin{proof}
Lemma~\ref{lem:orbit-normalization} and Proposition~\ref{prop:nonaxis-degenerate} transfer the problem to an admitted normalized representative and identify the exact six-regular simple domain; the inverse orbit isomorphism transfers the resulting decomposition back to $\EJ_\alpha$. Part i) is Corollary~\ref{cor:d1}, and part ii) is Theorems~\ref{thm:d2} and \ref{thm:d3}. For part iii), Lemma~\ref{lem:first-two-valid} proves that $H_1$ and $H_2$ are disjoint Hamiltonian cycles, and 6-regularity makes their complement $H_3$ a spanning 2-factor. If $d$ is even, Lemmas~\ref{lem:certified-orders} and \ref{lem:fine-block} show that its fine incidence graph is one cycle. If $d$ is odd, Lemmas~\ref{lem:certified-orders} and \ref{lem:odd-fine-block} give the same conclusion with the shifted terminal seam. The admissible-lift theorem, Theorem~\ref{thm:admissible}, therefore makes $H_3$ Hamiltonian in either parity. Corollary~\ref{cor:no-ratio-cases} makes this construction valid for every coprime normalized reduced ratio, with no additional arithmetic cases.
\end{proof}

\begin{Corollary}[Branchwise minimum local-switch count]\label{cor:optimal}
In the ordered local-switch model, the minimum total costs for $d=1$, $d=2$, $d=3$, and $d\ge4$ are, respectively,
\begin{equation}
 0,\qquad 4,\qquad 4,\qquad 2(d-1).
\label{eq:minimum-costs}
\end{equation}
For $d=3$ and $d\ge4$, each of the first two factors separately attains the lower bound $d-1$; for $d=2$, the optimal distribution is $1+3$ up to factor ordering.
\end{Corollary}

\begin{proof}
The $d=1$ branch needs no switches. Theorem~\ref{thm:d2-optimal} proves the value four for $d=2$, and Theorem~\ref{thm:d3} proves the value four for $d=3$. For $d\ge4$, both natural factors begin with $d$ components, so Lemma~\ref{lem:lowerbound} gives the total lower bound $2(d-1)$, attained by the two large-$d$ skeletons.
\end{proof}

\begin{Corollary}[End-to-end description size and complexity]\label{cor:complexity}
The complete decomposition is specified by the minimum totals $0$, $4$, $4$, and $2(d-1)$ local switches in the branches $d=1$, $d=2$, $d=3$, and $d\ge4$, respectively, with $H_3$ defined as the unused complement. After constant-size orbit normalization and gcd preprocessing, the compact certificate consists of $O(d)$ seed records and is generated using $O(d)$ word-RAM operations and $O(d)$ machine words. Its bit storage is $O(d\log N)$. Explicitly listing all three cycles takes $O(N)$ word-RAM operations and $O(N)$ output records.
\end{Corollary}

\begin{proof}
The counts follow from Corollary~\ref{cor:d1}, Theorems~\ref{thm:d2}--\ref{thm:d3}, and the two $d-1$ large-$d$ skeletons. After normalization and gcd preprocessing, each seed requires a constant number of arithmetic operations on $O(\log N)$-bit coordinates. Thus $O(d)$ seed records require $O(d)$ word-RAM operations and $O(d\log N)$ bits. Any explicit listing contains $3N$ undirected edges, so $O(N)$ word-RAM operations and output records are both sufficient and output-optimal up to constants.
\end{proof}

\section{Algorithm, Complexity, and Verification}

\begin{algorithm}[H]
\caption{Six-regular non-axis EJ EDHC construction}\label{alg:certified}
\begin{algorithmic}[1]
\REQUIRE $\alpha=a_0+b_0\rho\ne0$ with $a_0,b_0\in\ZZ$
\STATE form $\{R^j(a_0,b_0),R^jC(a_0,b_0):0\le j<6\}$ using \eqref{eq:orbit-matrices}
\STATE select a pair $(a,b)$ with $0\le a\le b$ (use fixed tie-breaking)
\IF{$a=0$}
 \STATE report ``axis-associated generator: outside the present theorem'' and stop
\ENDIF
\STATE set $N=a^2+ab+b^2$
\IF{$N=3$}
 \STATE report ``non-axis norm-three quotient: not six-regular simple'' and stop
\ENDIF
\STATE set $d=\gcd(a,b)$ and $(u,v)=(a/d,b/d)$
\IF{$d=1$}
 \STATE return the three natural unit-direction Hamiltonian cycles
\ELSIF{$d=2$}
 \STATE build $H_1$ with the anchor and $H_2$ from \eqref{eq:d2-seeds}
\ELSIF{$d=3$}
 \STATE build $H_1$ from \eqref{eq:H-skeleton} and $H_2$ from \eqref{eq:d3-seeds}
\ELSE
 \STATE build $H_1$ from \eqref{eq:H-skeleton}
 \STATE for $2\le i\le d-2$, set $P_i=K_i=1$ for even $i$ and $P_i=K_i=2$ for odd $i$
 \STATE build $H_2$ from $s_{B_0},s_{B_1}$ and \eqref{eq:universal-seed}
\ENDIF
\STATE set $H_3=E(\EJ_{a+b\rho})\setminus(E(H_1)\cup E(H_2))$
\STATE transport $H_1,H_2,H_3$ through the inverse orbit isomorphism to $\EJ_\alpha$
\STATE return $H_1,H_2,H_3$
\end{algorithmic}
\end{algorithm}

Lemma~\ref{lem:orbit-normalization} proves that Lines 1--3 inspect a constant-size orbit and always find a normalized representative. Together with Proposition~\ref{prop:nonaxis-degenerate}, the two explicit tests in Lines 3--9 reject exactly the axis-associated class and the unique normalized non-axis norm-three degeneration. Corollary~\ref{cor:complexity} gives the complete headline bound: the full decomposition uses $O(d)$ seed records---specifically $0$, $4$, $4$, or $2(d-1)$ switches in the four branches---and all $3N$ cycle edges are listed in $O(N)$ time. No seed search, reduced-ratio family classification, breadth-first search, or finite verification is required. An independent $O(N)$ verifier can still check degrees, connectedness, disjointness, cut coordinates, and fine incidence for reproducibility.

The construction and proofs are symbolic. Independent audits reconstruct the quotient graph, verify the three 2-regular connected factors, check zero pairwise overlap and union size $3N$, and validate cut orders and rank certificates. Table~\ref{tab:audit} gives only aggregate results. Supplementary File S1 states the exact deterministic protocol, finite audit domains, proof-critical successor dictionaries, compact coverage summaries, and representative full-graph records. Supplementary File S2 contains the final Python verification programs, deterministic case lists, machine-readable outputs, and exact commands needed to regenerate and cross-check the reported totals. In particular, a dictionary-free audit derives the fine-incidence multigraph directly from the cut orders and local connector endpoint rules, without using the manuscript's rank word, and verifies one $(4d-6)$-cycle for every $4\le d\le201$.

\begin{table}[H]
\centering
\caption{Aggregate deterministic verification. These computations are not used in Theorem~\ref{thm:main}.}
\label{tab:audit}
\small
\begin{adjustbox}{max width=\textwidth}
\begin{tabularx}{1.2\textwidth}{@{}X X r@{}}
\toprule
Audit family & Coverage & Result\\
\midrule
Small-branch arithmetic and reconstruction & 12,232 reduced ratios; 200 full $d=3$ quotients; six lower-budget $d=2$ cases & all pass\\
Universal rank, dictionary-free incidence, and odd endpoints & every $4\le d\le201$; 79,992 derived incidence edges; 10,094 expanded odd arrows & all pass\\
Large-branch full-quotient corpus & 980 cases, 200 reduced ratios, $N\le105{,}525$ & all pass\\
Axis-boundary classification & every $1\le n\le100$ & all pass\\
\bottomrule
\end{tabularx}
\end{adjustbox}
\end{table}

The supplementary PDF contains the detailed exceptional certificates, the complete even and odd successor dictionaries, the deterministic verification protocol, compact coverage summaries, and the relocated illustrative figures. Supplementary File S2 supplies the final standalone verifier, the dictionary-free fine-incidence graph audit, the exact deterministic case lists, all reported machine-readable outputs, and a one-command runner that also checks the totals and representative records printed in the paper. These materials are independent checks and are not required by any proof.

\section{Scope and Discussion}

Theorem~\ref{thm:main} is complete on the six-regular simple non-axis domain. Proposition~\ref{prop:axis-boundary} classifies the axis-associated boundary: $n=1,2$ are not six-regular simple, while $n\ge3$ is a distinct triangular-torus family with $Q=1$, so the alternating two-residue phase mechanism cannot apply there. Proposition~\ref{prop:nonaxis-degenerate} identifies the only non-axis simple-degree exception, $\alpha\sim1+\rho$, whose underlying simple graph is $C_3$. These are exact domain boundaries, not unresolved reduced-ratio cases inside the theorem.

The result concerns undirected simple six-regular EJ quotients with no failed vertices or links. Its minimum statements are with respect to the ordered local-switch model of Definition~\ref{def:local-cost}: $d=2$ has minimum total cost four, $d=3$ has minimum total cost four, and $d\ge4$ has minimum total cost $2(d-1)$. By Lemma~\ref{lem:switch-commutation}, chronological interleaving of the permitted switches does not enlarge this model. The paper does not claim optimality for broader operations that directly modify all three factors or exchange edges across factor copies; it also does not claim uniqueness, minimize geometric edge length under a drawing convention, or address fault-tolerant replacement after failures.

The proof separates component splicing, overlap control, and complement connectivity. The component trees force the switch lower bounds; anchor lock resolves the only forced non-diagonal overlap; and the rank certificate works on the actual diagonal path arcs. The equal-coordinate choice $P_i=K_i$ cancels the reduced ratio and modular inverse in the fine coordinate, leaving only the parity-dependent seam handled by the two rank schedules.

The compact certificate may reduce controller storage and simplify local verification: it is generated without search, stored as $O(d)$ seed records, and expanded only when the $O(N)$ cycle-edge output is required. This is a certificate-size statement, not a measured runtime advantage over the rectangular construction. The deterministic computational checks are independent stress tests; the theorem rests on the normalization, cancellation, cut-order, exhaustiveness, template-extension, and rank-bijection proofs. The dictionary-free incidence audit provides a separate combinatorial reconstruction of the fine-incidence cycle without reading the prescribed rank word.

\section{Conclusion}

Every six-regular simple non-axis EJ network admits the stated ratio-independent decomposition into three edge-disjoint Hamiltonian cycles. Within the canonical local-switch model of Definition~\ref{def:local-cost}, chronological interleaving is immaterial and the minimum total switch costs are $0$, $4$, $4$, and $2(d-1)$ for the admitted $d=1$, $d=2$, $d=3$, and $d\ge4$ branches, respectively. The equal-coordinate alternating lift removes all reduced-ratio family analysis, while parity-complete rank certificates prove that the unused complement is Hamiltonian. The construction uses $O(d)$ compact seed records and output-optimal $O(N)$ edge-listing time.

\section*{Supplementary Materials}
The following supporting information accompanies the submission: Supplementary File S1, \emph{Supplementary Proof Certificates and Verification Summary}, containing the expanded exceptional certificates, complete parity-specific successor dictionaries, deterministic verification protocol, coverage summaries, and representative full-graph records; and Supplementary File S2, \emph{Deterministic Reproducibility Package}, containing the final Python verifier, the dictionary-free fine-incidence graph audit, case lists, machine-readable outputs, checksums, and execution instructions.

\section*{Funding}
This research received no external funding.

\section*{Data Availability Statement}
This is a theoretical study. All proof-critical definitions, finite successor dictionaries, acceptance criteria, and aggregate audit results are included in the article and Supplementary File S1. The final deterministic Python verification programs, including the dictionary-free fine-incidence graph audit, exact case lists, and machine-readable audit outputs are provided in Supplementary File S2. These computational materials are independent checks and are not premises of the proofs.

\section*{Acknowledgments}
The author gratefully acknowledges Kuwait University and the Department of Computer Science, College of Science, Kuwait University, for providing the academic environment and institutional support that made this research possible.

\section*{Conflicts of Interest}
The author declares no conflict of interest.

\section*{Abbreviations}
The following abbreviations are used in this manuscript:
\begin{center}
\begin{tabular}{@{}ll@{}}
EJ & Eisenstein--Jacobi\\
EDHC & edge-disjoint Hamiltonian cycle\\
\end{tabular}
\end{center}


\begin{thebibliography}{999}

\bibitem{flahive2010}
Flahive, M.; Bose, B. The topology of Gaussian and Eisenstein--Jacobi interconnection networks. \emph{IEEE Trans. Parallel Distrib. Syst.} \textbf{2010}, \emph{21}, 1132--1142.

\bibitem{martinez2008ej}
Martinez, C.; Stafford, E.; Beivide, R.; Gabidulin, E.M. Modeling hexagonal constellations with Eisenstein--Jacobi graphs. \emph{Probl. Inf. Transm.} \textbf{2008}, \emph{44}, 1--11.

\bibitem{hussain2015}
Hussain, Z.A.; Bose, B.; Al-Dhelaan, A. Edge-disjoint Hamiltonian cycles in Eisenstein--Jacobi networks. \emph{J. Parallel Distrib. Comput.} \textbf{2015}, \emph{86}, 62--70.

\bibitem{hussain2022ist}
Hussain, Z.; AboElFotoh, H.; AlBdaiwi, B. Independent spanning trees in Eisenstein--Jacobi networks. \emph{J. Supercomput.} \textbf{2022}, \emph{78}, 12114--12135. https://doi.org/10.1007/s11227-022-04351-4.

\bibitem{awadh2023pan}
Awadh, M.; Hussain, Z.; Almansouri, H. Panconnectivity algorithm for Eisenstein--Jacobi networks. \emph{Kuwait J. Sci.} \textbf{2023}, \emph{50}, 485--491. https://doi.org/10.1016/j.kjs.2023.03.008.

\bibitem{hussain2024cist}
Hussain, Z.; AlAzemi, F.; AlBdaiwi, B. Completely independent spanning trees in Eisenstein--Jacobi networks. \emph{J. Supercomput.} \textbf{2024}, \emph{80}, 15105--15121. https://doi.org/10.1007/s11227-024-06042-8.

\bibitem{yang2023cayley}
Yang, D.-W.; Chen, X.-B. Symmetric property and edge-disjoint Hamiltonian cycles in Cayley graphs. \emph{Appl. Math. Comput.} \textbf{2023}, \emph{455}, 128113.

\bibitem{cheng2024balanced}
Cheng, B.-L.; Tan, J.M.; Pai, K.-J. Edge-disjoint Hamiltonian cycles in balanced hypercubes with applications to fault-tolerant data broadcasting. \emph{Appl. Sci.} \textbf{2024}, \emph{14}, 3232.

\bibitem{pai2026bcube}
Pai, K.-J.; Tan, J.M.; Hsu, L.-H. Constructing two edge-disjoint Hamiltonian cycles and two equal node-disjoint cycles in BCube data center networks for all-to-all broadcasting. \emph{Mathematics} \textbf{2026}, \emph{14}, 232.

\bibitem{albader2026gaussian}
Albader, B. A unified constant-time switch rule for constructing edge-disjoint Hamiltonian cycles in Gaussian networks. \emph{Mathematics} \textbf{2026}, \emph{14}, 2211.

\bibitem{albader2016}
Albader, B.; Bose, B. Edge-disjoint Hamiltonian cycles in Gaussian networks. \emph{IEEE Trans. Comput.} \textbf{2016}, \emph{65}, 315--321.

\bibitem{martinez2008gaussian}
Martinez, C.; Beivide, R.; Stafford, E.; Moreto, M.; Gabidulin, E.M. Modeling toroidal networks with the Gaussian integers. \emph{IEEE Trans. Comput.} \textbf{2008}, \emph{57}, 1046--1056.

\bibitem{bae2003}
Bae, M.M.; Bose, B. Edge-disjoint Hamiltonian cycles in $k$-ary $n$-cubes and hypercubes. \emph{IEEE Trans. Comput.} \textbf{2003}, \emph{52}, 1271--1284.

\bibitem{bae2000}
Bae, M.; Bose, B. Gray codes for torus and edge-disjoint Hamiltonian cycles. In \emph{Proceedings of the 14th International Parallel and Distributed Processing Symposium}, Cancun, Mexico, 1--5 May 2000; pp. 365--370.

\bibitem{bae2004}
Bae, M.; Bose, B.; AlBdaiwi, B.F. Edge-disjoint Hamiltonian cycles in two-dimensional torus. \emph{Int. J. Math. Math. Sci.} \textbf{2004}, \emph{2004}, 1299--1308.

\bibitem{latifi1993}
Latifi, S.; Zheng, S.-Q. On link-disjoint Hamiltonian cycles of torus networks. In \emph{Proceedings of IEEE Southeastcon}, Charlotte, NC, USA, 4--7 April 1993.

\bibitem{jha2012}
Jha, P.; Prasad, R. Hamiltonian decomposition of the rectangular twisted torus. \emph{IEEE Trans. Parallel Distrib. Syst.} \textbf{2012}, \emph{23}, 1504--1507.

\bibitem{rowley1991}
Rowley, R.; Bose, B. Edge-disjoint Hamiltonian cycles in De Bruijn networks. In \emph{Proceedings of the Sixth Distributed Memory Computing Conference}, Portland, OR, USA, 28 April--1 May 1991; pp. 707--709.

\bibitem{rowley1993}
Rowley, R.; Bose, B. On the number of arc-disjoint Hamiltonian circuits in the De Bruijn graphs. \emph{Parallel Process. Lett.} \textbf{1993}, \emph{3}, 375--382.

\bibitem{anantha2007}
Anantha, M.; Prasad, R.; AlBdaiwi, B.F. Mixed-radix Gray codes in Lee metric. \emph{IEEE Trans. Comput.} \textbf{2007}, \emph{56}, 1297--1307.

\bibitem{chen1981}
Chen, C.C.; Quimpo, N.F. On strongly Hamiltonian abelian group graphs. In \emph{Combinatorial Mathematics VIII}; Lecture Notes in Mathematics; Springer: Berlin/Heidelberg, Germany, 1981; Volume 884, pp. 23--34.

\bibitem{witte1984}
Witte, D.; Gallian, J.A. A survey: Hamiltonian cycles in Cayley graphs. \emph{Discrete Math.} \textbf{1984}, \emph{51}, 293--304.

\bibitem{bermond1989}
Bermond, J.-C.; Favaron, O.; Maheo, M. Hamiltonian decomposition of Cayley graphs of degree 4. \emph{J. Combin. Theory Ser. B} \textbf{1989}, \emph{46}, 142--153.

\bibitem{alspach2008}
Alspach, B. The wonderful Walecki construction. \emph{Bull. Inst. Combin. Appl.} \textbf{2008}, \emph{52}, 7--20.

\bibitem{duato2003}
Duato, J.; Yalamanchili, S.; Ni, L. \emph{Interconnection Networks: An Engineering Approach}; Morgan Kaufmann: San Francisco, CA, USA, 2003.

\bibitem{grama2003}
Grama, A.; Gupta, A.; Karypis, G.; Kumar, V. \emph{Introduction to Parallel Computing}, 2nd ed.; Addison-Wesley: Boston, MA, USA, 2003.

\bibitem{hardy1980}
Hardy, G.H.; Wright, E.M. \emph{An Introduction to the Theory of Numbers}, 5th ed.; Oxford University Press: Oxford, UK, 1980.

\bibitem{biggs1993}
Biggs, N. \emph{Algebraic Graph Theory}, 2nd ed.; Cambridge University Press: Cambridge, UK, 1993.

\bibitem{godsil2001}
Godsil, C.; Royle, G. \emph{Algebraic Graph Theory}; Springer: New York, NY, USA, 2001.

\bibitem{bondy2008}
Bondy, J.A.; Murty, U.S.R. \emph{Graph Theory}; Springer: London, UK, 2008.

\end{thebibliography}
\end{document}